\documentclass[a4paper]{amsart}
\usepackage[pdftex]{graphicx}
\usepackage[utf8]{inputenc}
\usepackage{amsmath,amsthm,amssymb, color, float, booktabs, tikz,natbib}
\usepackage{xurl}
\usepackage[hyperindex,breaklinks,hidelinks]{hyperref}
\usepackage{cleveref}

\crefname{appendix}{appendix}{appendices}
\Crefname{appendix}{Appendix}{Appendices}
\Crefname{thm}{Theorem}{Theorems}

\usetikzlibrary{arrows}
\usetikzlibrary{arrows.meta, positioning}

\newtheorem{lm}{Lemma}
\newtheorem{thm}{Theorem}
\newtheorem{prop}{Proposition}

\theoremstyle{definition}
\newtheorem{definition}{Definition}
\newtheorem{example}{Example}

\newcommand{\BM}{{\rm{BM}}}
\newcommand{\TTC}{{\rm{TTC}}}
\newcommand{\DAS}{{\rm{DA_S}}}

\newcommand{\DA}{{\rm{DA}}}
\newcommand{\EA}{{\rm{EADAM}}}

\newcommand\blfootnote[1]{%
  \begingroup
  \renewcommand\thefootnote{}\footnote{#1}%
  \addtocounter{footnote}{-1}%
  \endgroup
}

\makeatletter
\renewcommand{\subsection}{\@startsection{subsection}{2}{\z@}%
  {-3.25ex\@plus -1ex \@minus -.2ex}%
  {1.5ex \@plus .2ex}%
  {\normalfont\centering}}
\makeatother

\begin{document}
\thispagestyle{empty}

\title{Credibility in school choice}

\author[Camilo J. Sirguiado]{Camilo J. Sirguiado$^{\dagger}$}

\author[Jiarui Xie]{Jiarui Xie$^{\ddagger}$}

\begin{abstract}
In centralized school choice, a designer who announces a mechanism may deviate from it to favor some students without being detected. A mechanism is credible if it admits no such deviation. We study this credibility problem when students do not know others' reports and may observe only part of the assignment. 
Credibility is demanding: among common school choice mechanisms, only deferred acceptance is credible, and only when students observe enough of the assignment.
We therefore rank mechanisms by their credibility. 
Deferred acceptance is more credible than any other stable mechanism and strictly more credible than the Boston mechanism, regardless of how much of the assignment students observe.
Its ranking relative to top trading cycles and efficiency-adjusted deferred acceptance depends on disclosure: deferred acceptance is strictly more credible when students observe the entire assignment, whereas the ranking reverses when students observe only their own assignment.
\end{abstract}

\maketitle

\blfootnote{We thank Mohammad Akbarpour, Itai Ashlagi, Yuichiro Kamada, Shengwu Li, and Andrzej Skrzypacz for helpful comments and suggestions.}
\blfootnote{$^{\dagger}$\,Department of Economics, University of California, Berkeley. Contact: \href{mailto:camijsir@berkeley.edu}{camijsir@berkeley.edu}}
\blfootnote{$^{\ddagger}$\,Graduate School of Business, Stanford University. Contact: \href{mailto:jrxie@stanford.edu}{jrxie@stanford.edu}}

\section{Introduction}

Over the past two decades, many countries have adopted centralized school choice systems. \citet{neilson2024rise} documents that 60\% of countries with more than two million residents use at least one centralized system to assign students to schools. These systems can reduce congestion and unraveling, and thus improve the efficiency and fairness of assignments \citep{abdulkadirouglu2017welfare}. These benefits, however, depend on families trusting the market designer to follow the announced assignment rule. If the designer can change assignments without being detected, the centralized system has a credibility problem.

Credibility concerns arise in school admissions around the world. One source of these concerns is deliberate favoritism by officials. In Huaihua, China, a government official arranged for the admission of more than 200 ineligible students to public primary schools between 2017 and 2020.\footnote{Xinyu Yang, \textit{A Hard Drive Exposes a School-Admissions Corruption Ring}, China Youth Daily, August 25, 2021, \url{https://news.cyol.com/gb/articles/2021-08/25/content_Q7xzqtjPW.html}, accessed August 25, 2026.} In 2015, a university official was charged with accepting bribes between 2005 and 2013 to help more than 40 students enter a prestigious Chinese university or receive other preferential treatment.\footnote{Ran Shen, \textit{Renmin University's Rongsheng Cai Stands Trial, Admits Accepting More Than 20 Million Yuan in Bribes}, China News Service, December 3, 2015, \url{https://www.chinanews.com/sh/2015/12-03/7654785.shtml}, accessed August 25, 2026.} Similar favoritism occurred in the United States. In Washington, D.C., school officials approved seven requests from politically connected families to bypass the public-school lottery in 2015.\footnote{Peter Jamison and Aaron C. Davis, \textit{Secret Report Shows ``Special'' Treatment for Public Officials in D.C. School Lottery}, The Washington Post, May 17, 2017, \url{https://www.washingtonpost.com/local/dc-politics/secret-report-shows-special-treatment-for-public-officials-in-dc-school-lottery/2017/05/17/55b0b1fc-3a82-11e7-8854-21f359183e8c_story.html}, accessed August 25, 2026.} 
After this misconduct became public, the new school chancellor bypassed the same process in 2017 to place his daughter.\footnote{Perry Stein and Peter Jamison, \textit{Deputy Mayor Resigns after Allowing Chancellor to Bypass School Lottery}, The Washington Post, February 16, 2018, \url{https://www.washingtonpost.com/news/education/wp/2018/02/16/auto-draft/}, accessed August 25, 2026.}

Implementation errors are another source of credibility concerns. In Boston, a grading-system error affected more than 100 students applying to selective public schools in 2019 and 2020. Some were wrongly denied admission, while others were admitted even though they were ineligible.\footnote{Naomi Martin, \textit{Dozens of Students Were Not Admitted to Boston Exam Schools Because of an Error}, The Boston Globe, August 31, 2020, \url{https://www.bostonglobe.com/2020/08/31/metro/boston-public-schools-announces-error-exam-school-admissions-that-kept-dozens-out-recent-years/}, accessed August 25, 2026.} 
Even after the earlier error was discovered, Boston officials made another mistake in 2023: dozens more students were incorrectly notified about their eligibility to apply to selective public schools.\footnote{Max Larkin, \textit{Boston Public Schools Admits It Sent Incorrect Exam School Eligibility Letters}, WBUR, April 12, 2023, \url{https://www.wbur.org/news/2023/04/12/boston-public-schools-exam-school-letter-errors}, accessed August 25, 2026.} In Gothenburg, Sweden, a new computerized assignment system calculated school proximity incorrectly in 2020, resulting in nearly 450 incorrect assignments.\footnote{Andreas Varas Hernandez, \textit{Gothenburg's School Choice Internally Audited---Report Finds Several Errors}, SVT Nyheter, June 18, 2020, \url{https://www.svt.se/nyheter/lokalt/vast/skolvalet-i-goteborg-har-intergranskats}; Lars Wiklund, \textit{School Board Criticized after Gothenburg School-Choice Chaos}, SVT Nyheter, March 17, 2021, \url{https://www.svt.se/nyheter/lokalt/vast/grundskolenamnd-kritiseras-efter-skolvalskaoset-i-goteborg}, accessed August 25, 2026.}

In each of these examples, families may suspect that some students benefited from undetected deviations from the announced rule. How serious this credibility problem is depends on the assignment mechanism and the information disclosed after the assignment. In this paper, we ask which of the mechanisms commonly used in school choice best addresses the credibility problem, and how the answer depends on what families observe about the final assignment. Indeed, changes in the amount of disclosure can reverse the ranking of mechanisms.

We study this question in the school choice model of \citet{abdulkadirouglu2003school}. Following \citet{moller2026transparent}, the designer announces a mechanism but cannot fully commit to using it. Students know the mechanism, priorities, and quotas, but not the preferences reported by other students.\footnote{In practice, students may also observe some preferences reported by other students. We do not consider this possibility in this paper because privacy concerns make such disclosure unlikely.} After the assignment, each student observes her own assignment and the assignments of a subset of other students; we call this information her \emph{observation}.\footnote{Unlike \citet{moller2026transparent}, who assumes that each student observes the entire assignment, we allow each student to observe only part of it.} A student can \emph{explain} her observation if some reports by the other students would lead the announced mechanism to produce the assignments she observes. 
A matching is a \emph{safe deviation} if it differs from the outcome that the announced mechanism would produce under the students' reports, all students can explain their observations, and that outcome does not Pareto dominate it.\footnote{Unlike \citet{moller2026transparent}, we require that a deviation benefit at least one student. We add this requirement to model a designer who deviates to favor particular students.} A mechanism is \emph{credible} if it admits no safe deviation.

Credibility is a demanding requirement. The Boston mechanism and the top trading cycles mechanism are not credible even when every student observes the entire assignment, and any stable mechanism can admit a safe deviation when every student except one observes the entire assignment.\footnote{These claims follow from the examples in the proofs of \Cref{thm:DA-stable,thm:DA-BM,thm:DA-TTC}.}
As a binary property, credibility cannot rank mechanisms that both admit safe deviations. We therefore compare mechanisms environment by environment. A mechanism $\varphi$ is \emph{more credible} than a mechanism $\psi$ if, in every environment, $\psi$ admits a safe deviation whenever $\varphi$ does. 
The mechanism $\varphi$ is \emph{strictly more credible} than $\psi$ if, in addition, $\psi$ is not more credible than $\varphi$.

We use this comparison to rank four prominent school choice mechanisms. We begin with the two mechanisms most widely used in practice: the student-proposing deferred acceptance mechanism ($\DA$; \citealp{gale1962college}) and the Boston mechanism ($\BM$; \citealp{abdulkadirouglu2003school}).\footnote{$\DA$ is the most widely used mechanism in higher education, while $\BM$ is the most common in primary and secondary education \citep{neilson2024rise}.} $\DA$ selects the student-optimal stable matching and is strategy-proof, but it is not Pareto efficient. $\BM$ is Pareto efficient but neither stable nor strategy-proof.
We also study two efficient mechanisms: the top trading cycles mechanism ($\TTC$; \citealp{abdulkadirouglu2003school}), which is strategy-proof but not stable, and the efficiency-adjusted deferred acceptance mechanism ($\EA$; \citealp{kesten2010school}), which weakly Pareto dominates $\DA$ and is close to being stable \citep{dougan2021minimally} but not strategy-proof.\footnote{Throughout the paper, $\EA$ denotes the efficiency-adjusted deferred acceptance mechanism when all students waive their priorities.}

We have three main results. The first compares $\DA$ with other stable mechanisms. We show that $\DA$ is more credible than any other stable mechanism, regardless of what students observe. This result identifies a new advantage of $\DA$ within the class of stable mechanisms. Moreover, when students observe the entire assignment, $\DA$ is the only credible stable mechanism. Thus, if the designer uses $\DA$, disclosing the entire assignment is enough to ensure that any deviation is detected. As the remaining results show, the other mechanisms we study do not share this property.

Our second result shows that $\DA$ is strictly more credible than $\BM$, regardless of what students observe. The intuition comes from stability: $\DA$ is stable, whereas $\BM$ is not. Under $\DA$, if a student observes a lower-priority student assigned to a school that she prefers to her own assignment, she knows that the announced mechanism could not have produced the matching. $\BM$ does not impose the same restriction, so students can explain more matchings, giving the designer more opportunities for safe deviations.

Our third result compares $\DA$ with $\TTC$ and $\EA$. When students observe the entire assignment, $\DA$ is strictly more credible than both mechanisms. However, this ranking reverses when students observe only their own assignment: $\TTC$ and $\EA$ are strictly more credible than $\DA$. To the best of our knowledge, centralized school choice systems in practice are closer to own-assignment observation than to full observation. Thus, our result suggests that $\TTC$ or $\EA$ may be preferable to $\DA$ in practice when credibility is a concern.

Our main results use the standard school choice model, but several features of centralized school choice systems may affect credibility and the resulting mechanism ranking. Therefore, we study three such extensions.

First, centralized systems often limit the number of schools that a student can report as acceptable. In fact, more than 80\% of centralized systems in primary and secondary education impose such a restriction \citep{neilson2024rise}. Although this restriction changes assignments and may destroy stability and strategy-proofness, all our main credibility rankings continue to hold when students can report at least two schools.

Second, admissions systems frequently rely on weak priorities and lotteries. If students do not observe the realized tie-breaking rank, they cannot verify whether the designer followed the announced lottery rule, as illustrated by the D.C. lottery cases. We assume that a student observes the weak priority profile and her own realized rank, but not the ranks of other students. Our credibility rankings continue to hold with one exception: when students observe the entire assignment, $\DA$ is no longer more credible than $\TTC$ or $\EA$. Moreover, when all schools are indifferent among students before tie-breaking, $\TTC$ and $\EA$ are strictly more credible than $\DA$, regardless of how much of the assignment each student observes.

Third, different types of students are often admitted under different rules. We compare two such affirmative-action policies: minority reserves and majority quotas. When students observe the entire assignment, deferred acceptance is credible under either policy, so credibility does not rank them. A ranking arises when each student observes her own assignment and the assignments of any subset of minority students. In this case, deferred acceptance with minority reserves is strictly more credible than deferred acceptance with majority quotas when, at every school, the minority reserve and the majority quota sum to the school's quota. Since minority reserves also weakly Pareto dominate the corresponding majority quotas \citep{hafalir2013effective}, minority reserves dominate majority quotas in both welfare and credibility.

The rest of the paper is organized as follows. \Cref{sec:literature} discusses the related literature. \Cref{sec:model} introduces the model and our notions of safe deviation and credibility. \Cref{sec:main} presents the main results on the credibility rankings of school choice mechanisms. \Cref{sec:practice} considers three extensions: constrained choice, weak priorities, and affirmative action. \Cref{sec:conclusion} concludes. All proofs are relegated to \Cref{sec:proof}, while \Cref{sec:example} provides additional examples.

\subsection{Related literature}\label{sec:literature}

Our paper builds on \citet{akbarpour2020credible}, who introduce credibility in auctions. They define a mechanism as credible if the auctioneer has no safe deviation that increases her revenue. Since revenue is not the designer's objective in school choice, we instead focus on student welfare and require that a safe deviation not be Pareto dominated by the original outcome.

Within the matching literature, this paper is most closely related to \citet{moller2026transparent}, which introduces safe deviations in matching and defines a mechanism as transparent if it has no safe deviations. Our model differs in two ways. First, whereas \citet{moller2026transparent} requires only that a safe deviation have an innocent explanation, we also require that the original outcome of the mechanism not Pareto dominate it. 
Thus, under full observation, $\DA$ is credible in our model but not transparent in \citet{moller2026transparent}.
Second, whereas \citet{moller2026transparent} assumes that each student observes the entire assignment, we allow each student to observe only part of it.

Both \citet{akbarpour2020credible} and \citet{moller2026transparent} study whether a given mechanism has any safe deviation in a fixed environment. Our approach is different: we compare mechanisms environment by environment, which allows us to rank them even when neither is credible. A related comparison is provided by \citet{grigoryan2023theory}, who measure auditability by the minimum number of students whose combined information is sufficient to detect a deviation. Rather than varying the number of students required for auditing, we fix each student's information and ask which of two mechanisms is more credible under the same environment.

A growing literature studies the role of information in matching markets. One strand studies how assignment mechanisms and admission transparency affect information acquisition and student welfare \citep{artemov2021assignment,koh2026transparency}. A second strand studies how information disclosure allows students to verify matching outcomes.
Within this strand, \citet{hakimov2026improving} design disclosure procedures for students to verify school assignments. \citet{gonczarowski2024structural} study the communication complexity of representing and verifying the outcomes of $\DA$ and $\TTC$. \citet{pycia2023ordinal} study auditability and ordinal verification in discrete mechanism design. 
Related to our affirmative-action application, \citet{hu2026verifiable} characterize mechanisms that students can verify when cutoffs are disclosed. Our paper differs from both strands. Rather than studying how disclosure affects student behavior or how to design a disclosure rule, we fix what each student observes and compare the credibility of mechanisms.

Our method of comparing mechanisms environment by environment follows \citet{pathak2013school}, who introduce this method to rank mechanisms by manipulability. \citet{chen2017chinese} use the same method to compare school choice mechanisms by manipulability and stability.
Our results add credibility as another dimension of comparison.

Our work is also related to the literature on constrained school choice \citep{haeringer2009constrained,calsamiglia2010constrained}, weak priorities and tie-breaking \citep{erdil2008s,abdulkadirouglu2009strategy}, and affirmative action and controlled choice \citep{kojima2012school,hafalir2013effective,echenique2015control}. In each setting, we introduce credibility as an additional criterion for comparing mechanisms.

\section{Model}\label{sec:model}

Let $I$ be a set of students and $S$ a set of schools such that $|I|\geq 3$ and $|S|\geq 2$. Each school $s\in S$ has a strict priority order $\succ_s$ defined over $I$ and a quota $q_s\geq 1$. Each student $i\in I$ has a complete, transitive, and strict preference relation $P_i$ defined over $S\cup \{\emptyset\}$, where $\emptyset$ represents an outside option. Let $R_i$ be the weak preference relation induced by $P_i$. A school $s$ is \textbf{acceptable} to student $i$ if $s \mathrel{P_i} \emptyset$. Given $\succ=(\succ_s)_{s\in S}$, $q=(q_s)_{s\in S}$, and $P=(P_i)_{i\in I}$, we refer to $[I, S, \succ, q, P]$ as a \textbf{school choice problem}, or simply a problem.

A \textbf{matching} is a mapping $\mu: I\cup S \to I\cup S\cup \{\emptyset\}$ such that for any $(i,s)\in I\times S$ we have $\mu(i)\in S\cup \{\emptyset\}$, $\mu(s)\subseteq I$, $|\mu(s)|\leq q_s$, and $\mu(i)=s$ if and only if $i\in \mu(s)$.
Let $\mathcal{M}$ be the set of possible matchings between $I$ and $S$. A matching $\mu \in \mathcal{M}$ is \textbf{individually rational} if no student $i$ prefers $\emptyset$ to $\mu(i)$. Moreover, $\mu$ is \textbf{stable} if it is individually rational and there is no pair $(i,s)\in I\times S$ such that $s \mathrel{P_i} \mu(i)$ and either $|\mu(s)|<q_s$ or $i \succ_s j$ for some $j\in \mu(s)$. A matching $\mu\in \mathcal{M}$ is \textbf{Pareto dominated} by $\eta\in \mathcal{M}$ if $\eta(i) \mathrel{R_i} \mu(i)$ for all $i\in I$ and $\eta(i) \mathrel{P_i} \mu(i)$ for some $i\in I$. Finally, $\mu$ is \textbf{Pareto efficient} if it is not Pareto dominated by any $\eta\in \mathcal{M}$.

A \textbf{mechanism} $\varphi$ is a function that maps each problem to a matching. When $I$, $S$, $\succ$, and $q$ are fixed, we use only the students' preference profile as the argument of $\varphi$. Thus, we write $\varphi:\mathcal{P}\to\mathcal{M}$, where $\mathcal{P}$ is the set of possible preference profiles, and use $\varphi[P]$ as shorthand for $\varphi[I,S,\succ,q,P]$. 
Given a preference profile $P$ and a student $i$, let $P_{-i}=(P_j)_{j\in I\setminus \{i\}}$. Denote by $\varphi[P_i, P_{-i}](i)$ the assignment of student $i$ when she reports $P_i$ and the other students report $P_{-i}$. Similarly, $\varphi[P](s)$ denotes the set of students assigned to school $s$ under the preference profile $P$. A mechanism $\varphi$ is \textbf{stable (efficient)} if $\varphi[P]$ is stable (efficient) under every $P\in \mathcal{P}$. 

A report $P'_i$ of student $i\in I$ is a (profitable) manipulation of the mechanism $\varphi$ under the preference profile $P$ if $\varphi[P'_i, P_{-i}](i) \mathrel{P_i} \varphi[P](i)$. 
A mechanism $\varphi$ is \textbf{strategy-proof} if no student has a manipulation under any $P\in \mathcal{P}$. 

Students have only partial information about the school choice problem $[I, S, \succ, q, P]$ and the matching implemented by the designer. In particular, each student knows the sets of students $I$ and schools $S$, the priority profile $\succ$, the quota profile $q$, and her own preference $P_i$, but not the preferences reported by other students. In addition, she observes only part of the assignment. More precisely, consider an \textbf{observation structure} $\Gamma: I\to 2^I$ such that $i\in\Gamma(i)$ for each student $i$. Given the observation structure $\Gamma$ and a matching $\mu \in \mathcal{M}$, let $o_i(\mu, \Gamma) = \{(j, s)\in \Gamma(i)\times (S\cup\{\emptyset\}): \mu(j)=s\}$ be student $i$'s \textbf{observation} under matching $\mu$. In words, each student $i\in I$ observes only the assignments of the students in $\Gamma(i)$. Thus, if $\Gamma(i)=\{i\}$ for every student $i$, we say that \textbf{students observe their own assignment}; if $\Gamma(i)=I$ for every student $i$, we say that \textbf{students observe the entire assignment}.

\begin{definition}\label{def:safe}
Given an observation structure $\Gamma$, a matching $\mu$ is a \textbf{safe deviation} of a mechanism $\varphi$ under $P\in \mathcal{P}$ if the following conditions hold:
\begin{itemize}
  \item[(i)] $\varphi[P]\neq \mu$.
  \item[(ii)] For any $i\in I$, there exists $P' \in \mathcal{P}$ such that $o_i(\varphi[P_i, P'_{-i}], \Gamma)=o_i(\mu, \Gamma)$.
  \item[(iii)] $\mu$ is not Pareto dominated by $\varphi[P]$.
\end{itemize}
\end{definition}
In words, condition~(i) requires the designer to change the outcome. Condition~(ii) requires each student to be able to \emph{explain} her observation under $\mu$ as the announced mechanism's outcome for some reports by the other students. We refer to the corresponding profile $P_{-i}'$ as an \textbf{innocent explanation} for student $i$'s observation under $\mu$. Condition~(iii) requires the deviation to benefit at least one student.

Given an observation structure $\Gamma$, a mechanism $\varphi$ is \textbf{credible} if it has no safe deviation under any $P\in \mathcal{P}$. 

\begin{definition}
Given an observation structure $\Gamma$, a mechanism $\varphi$ is \textbf{more credible} than mechanism $\psi$ if, for any $P \in \mathcal{P}$, if $\varphi$ has a safe deviation under $P$, then $\psi$ also has a safe deviation under $P$.
\end{definition}
A mechanism $\varphi$ is \textbf{strictly more credible} than mechanism $\psi$ if $\varphi$ is more credible than $\psi$ and there exists a school choice problem $[I, S, \succ, q, P]$ under which $\psi$ has a safe deviation but $\varphi$ does not.

\section{Main results}\label{sec:main}

This section compares the credibility of $\DA$ \citep{gale1962college}, $\BM$ \citep{abdulkadirouglu2003school}, $\TTC$ \citep{abdulkadirouglu2003school}, $\EA$ \citep{kesten2010school}, and the class of stable mechanisms. We begin with stability, then turn to Pareto-efficient mechanisms and show how the comparison depends on what students observe.

Stability is a central objective in school choice since it respects publicly known priorities. Among stable mechanisms, $\DA$ selects the student-optimal stable matching \citep{gale1962college}. Our first result shows that this optimality makes $\DA$ more credible than any other stable mechanism, regardless of what students observe.

\begin{thm}\label{thm:DA-stable}
Fix any observation structure $\Gamma$.
\begin{itemize}
  \item[(i)] $\DA$ is more credible than any other stable mechanism.
  \item[(ii)] If $\Gamma(i)=I$ for any $i\in I$, $\DA$ is the only stable and credible mechanism.
\end{itemize}
\end{thm}

We provide an intuition for \Cref{thm:DA-stable}. First, consider condition~(ii) of \Cref{def:safe}.
Fix a student $i$ and a matching $\mu$. If $i$ observes either an empty seat at a school she prefers to $\mu(i)$ or a lower-priority student assigned there, then no reports by other students can make this observation consistent with stability. Conversely, if the observed assignments satisfy these restrictions and enough unobserved higher-priority students can fill the remaining seats at every school that $i$ prefers, then we can choose those students' preferences so that any stable mechanism produces the assignments that $i$ observes. Thus, whether $\mu$ satisfies condition~(ii) is the same for all stable mechanisms.

Stable mechanisms differ only under condition~(iii) of \Cref{def:safe}, which compares $\mu$ with the original outcome of the mechanism. Since $\DA[P]$ weakly Pareto dominates the outcome of any other stable mechanism, any matching that is not Pareto dominated by $\DA[P]$ cannot be Pareto dominated by the outcome of another stable mechanism. Thus, any safe deviation of $\DA$ is also a safe deviation of any other stable mechanism, which proves the weak ranking in part~(i) of \Cref{thm:DA-stable}.

For part~(ii) of \Cref{thm:DA-stable}, if students observe the entire assignment, condition~(ii) of \Cref{def:safe} requires the deviation to be stable. Student optimality then implies that $\DA[P]$ Pareto dominates any other deviation, so $\DA$ is credible and no other stable mechanism is credible.

\Cref{thm:DA-stable} therefore provides a credibility-based reason to prefer the student-proposing deferred acceptance mechanism over the school-proposing version. Although part~(i) provides only a weak comparison between these two mechanisms, \Cref{exp:school-proposing} in \Cref{sec:example} shows that $\DA$ is strictly more credible than the school-proposing deferred acceptance mechanism under any observation structure.

Pareto efficiency is another central objective in school choice settings. Since stability and Pareto efficiency are incompatible under general priority structures \citep{balinski1999tale}, the credibility comparisons among stable mechanisms leave open how $\DA$ compares with efficient mechanisms. While Pareto efficiency restricts the set of candidate deviations that can satisfy condition~(iii) of \Cref{def:safe}, an inefficient mechanism can still be more credible than an efficient one.
To show this, we first compare the two mechanisms most widely used in centralized school choice worldwide: $\DA$ and the Boston mechanism.

\begin{thm}\label{thm:DA-BM}
Fix any observation structure $\Gamma$. Then $\DA$ is strictly more credible than $\BM$.
\end{thm}

We provide an intuition for \Cref{thm:DA-BM}. 
As discussed after \Cref{thm:DA-stable}, the constraints on safe deviations under $\DA$ arise from stability and observable priorities. In contrast, $\BM$ makes much weaker use of priorities: a school compares students by priority only when the number of applicants in the same round exceeds its remaining capacity. Thus, priorities impose fewer restrictions on safe deviations under $\BM$. Pareto efficiency works in the opposite direction by restricting safe deviations under $\BM$, but this restriction does not offset the weaker use of priorities. Therefore, $\BM$ is less credible than $\DA$.

\Cref{thm:DA-BM} shows that replacing $\BM$ with $\DA$ improves credibility under any observation structure. Previous work shows that $\BM$ is neither strategy-proof nor stable \citep{abdulkadirouglu2003school}, and these concerns have led school districts to replace it with $\DA$. Boston, for example, adopted $\DA$ in 2005, partly because strategic reporting under $\BM$ harmed families that did not understand the mechanism \citep{abdulkadiroglu2006changing}. Credibility therefore provides an additional reason for such reforms.

We next compare $\DA$ with two Pareto-efficient mechanisms: $\TTC$ and $\EA$. $\TTC$ is strategy-proof, but its outcome need not Pareto dominate the outcome of $\DA$. By contrast, $\EA$ is not strategy-proof, but it weakly Pareto dominates $\DA$. The next result shows that their credibility rankings relative to $\DA$ depend on what students observe.

\begin{thm}\label{thm:DA-TTC}
Fix any observation structure $\Gamma$.
\begin{itemize}
  \item[(i)] If $\Gamma(i)=I$ for any $i\in I$, then $\DA$ is strictly more credible than $\TTC$ and $\EA$.
  \item[(ii)] If $\Gamma(i)=\{i\}$ for any $i\in I$, then $\TTC$ and $\EA$ are strictly more credible than $\DA$.
\end{itemize}
\end{thm}

The intuition for \Cref{thm:DA-TTC} differs from the Boston case, because priorities impose additional restrictions on safe deviations under $\TTC$ and $\EA$. Under $\TTC$, priorities determine the order in which schools point to students. Under $\EA$, priorities restrict outcomes indirectly since $\EA$ must weakly Pareto dominate $\DA$. Whether these restrictions offset the restrictions imposed by stability under $\DA$ depends on what students observe.

Part~(i) considers the case in which students observe the entire assignment. \Cref{thm:DA-stable} then implies that $\DA$ is credible. By contrast, neither $\TTC$ nor $\EA$ is stable, so an unstable matching can still have an innocent explanation under either mechanism, leaving room for safe deviations. Indeed, we show in the proof that neither mechanism is credible in this case. Thus, $\DA$ is strictly more credible than both mechanisms.

Part~(ii) considers the case in which students observe only their own assignment. Here, students can no longer compare their priorities with those of students assigned to preferred schools. Stability therefore imposes fewer restrictions in this case, while Pareto efficiency makes condition~(iii) of \Cref{def:safe} more restrictive under $\TTC$ and $\EA$ than under $\DA$. As a result, both mechanisms are strictly more credible than $\DA$. 

\Cref{thm:DA-TTC} shows that the tradeoff between efficiency and credibility depends on what students observe. 
Under full observation, $\DA$ is credible, while $\TTC$ and $\EA$ are not. Under own-assignment observation, the ranking reverses: $\TTC$ and $\EA$ are strictly more credible than $\DA$, so the designer need not trade off efficiency against credibility. Because centralized systems are closer in practice to own-assignment observation than to full observation, the own-assignment ranking is especially relevant.
However, the comparison is less robust between these two extremes. When students observe only part of the assignment, both $\DA$ and $\BM$ can be incomparable with each of the two efficient mechanisms (\Cref{exp:DA-TTC-partial,exp:BM-TTC} in \Cref{sec:example}).

\section{School choice in practice}\label{sec:practice}

In practice, school choice mechanisms face multiple constraints. In this section, we consider three common departures from the model in \Cref{sec:model}: limits on the length of preference lists (\Cref{sec:constrained}), weak priorities with tie-breaking (\Cref{sec:weak}), and affirmative-action policies (\Cref{sec:affirmative}). For each setting, we study how our results are affected.

\subsection{Constrained school choice}\label{sec:constrained}

Centralized admissions systems often limit the number of choices that students can report. Chicago high school applicants may apply to up to twenty programs.\footnote{For more on Chicago high school admissions, see \url{https://www.cps.edu/gocps/high-school/explore/choice-programs/}, accessed August 25, 2026.} In London, students can apply to at most six schools in primary and secondary school admissions.\footnote{For more on London primary and secondary school admissions, see \url{https://explore-education-statistics.service.gov.uk/find-statistics/primary-and-secondary-school-applications-and-offers/2025-26}, accessed August 25, 2026.} Beijing's 2025 college admissions allow only thirty choices.\footnote{For more on Beijing college admissions, see \url{https://www.beijing.gov.cn/fuwu/bmfw/sy/jrts/202506/t20250610_4109714.html}, accessed August 25, 2026.} Under this list-length constraint, $\DA$ is no longer strategy-proof \citep{haeringer2009constrained}, raising the question of whether the constraint also changes the credibility rankings in \Cref{sec:main}. To answer this question, we first define the constrained version of each mechanism.

Suppose that students may report at most $k\geq 1$ acceptable schools. If $n_i$ is the number of schools acceptable under $P_i$, let $P_i^k$ preserve the first $\min\{k,n_i\}$ acceptable schools under $P_i$ and rank any other school below $\emptyset$. The \textbf{constrained version} of a mechanism $\varphi$ is the mechanism $\varphi^k$ defined by $\varphi^k[P]=\varphi[P^k]$ for any $P\in\mathcal{P}$. A mechanism is \textbf{constrained stable} or \textbf{constrained efficient} if its outcome is stable or efficient under $P^k$ for any $P\in\mathcal{P}$.

This list-length constraint affects credibility in two ways. On the one hand, by restricting the reports available in an innocent explanation, it can make a deviation easier to detect. On the other hand, truncation can leave a student unmatched or assign her to a less preferred school. These worse outcomes then allow more deviations to satisfy condition~(iii) of \Cref{def:safe}. Therefore, the results from the unrestricted model do not apply here automatically.

That said, we show that a list-length constraint does not change the credibility rankings when students can report at least two schools.

\begin{prop}\label{prop:constrained}
Fix any observation structure $\Gamma$ and constraint $k\geq 2$.
\begin{itemize}
  \item[(i)] $\DA^k$ is more credible than the constrained version of any other stable mechanism.
  \item[(ii)] $\DA^k$ is strictly more credible than $\BM^k$.
  \item[(iii)] If $\Gamma(i)=I$ for any $i\in I$, $\DA^k$ is the only constrained stable and credible mechanism. Moreover, $\DA^k$ is strictly more credible than $\TTC^k$ and $\EA^k$.
  \item[(iv)] If $\Gamma(i)=\{i\}$ for any $i\in I$, then $\TTC^k$ and $\EA^k$ are strictly more credible than $\DA^k$.
\end{itemize}
\end{prop}

In the proof of \Cref{prop:constrained}, parts~(i),~(ii), and~(iv) rely on the idea of \emph{simple explanations}: each explanation can be constructed so that every other student reports at most one acceptable school. Stable mechanisms and $\BM$ admit simple explanations (see \Cref{lm:simple}). Moreover, if a mechanism is more credible than one that admits simple explanations, then the constrained version of the first is more credible than the constrained version of the second (see \Cref{lm:constrained}). This establishes the weak comparisons in parts~(i),~(ii), and~(iv), which hold for any constraint $k\geq 1$.

To establish the strict comparisons, we rely on examples in which some students find at least two acceptable schools. When $k = 1$, the constrained versions of all stable mechanisms, $\BM$, and $\EA$ coincide on every preference profile. Therefore, no strict ranking among them is possible, which explains why \Cref{prop:constrained} requires $k \geq 2$. For part~(iii), since $\DA^k$ is credible under full observation, it suffices to show one such example in which $\TTC^k$ and $\EA^k$ are not credible.

A list-length constraint can affect assignments, welfare, and incentives, but \Cref{prop:constrained} shows that it does not change the credibility rankings when $k\geq 2$. Therefore, a designer can choose the list length for other reasons without affecting the ranking of mechanisms implied by credibility.

\subsection{Weak priorities}\label{sec:weak}

Schools often place students into priority groups rather than rank them strictly. These groups may depend on sibling status, residence, or program eligibility, with lotteries used to break ties within each group. 
Boston, New York City, Washington, D.C., and Beijing use such groups and lotteries to assign public school seats.\footnote{For more on Boston public school admissions, see \url{https://www.bostonpublicschools.org/enrollment/welcome-services/student-assignment-policy}, accessed August 25, 2026. For more on New York City public school admissions, see \url{https://www.schools.nyc.gov/enrollment/enroll-grade-by-grade/how-students-get-offers-to-doe-public-schools/random-numbers-in-admissions}, accessed August 25, 2026. For more on Washington, D.C., public school admissions, see \url{https://www.myschooldc.org/faq/my-school-dc-lottery-application-and-common-lottery}, accessed August 25, 2026. For more on Beijing public school admissions, see \url{https://www.bjdch.gov.cn/zwgk/zcwj2024/202604/t20260428_4620458.html}, accessed August 25, 2026.}
These lotteries create an additional information problem, as a student may know the priority group and her own lottery position without knowing the lottery positions of other students.

We capture this problem by assuming that each student observes the weak priority profile and her own rank under the realized tie-breaking, but not the realized ranks of other students. Let $\succsim_s$ denote the weak priority of school $s$. A strict priority profile $\succ$ is a \textbf{tie-breaking} of $\succsim$ if it preserves any strict comparison under $\succsim_s$ at every school. That is, $i \succ_s j$ implies $i \succsim_s j$ for any $i, j \in I$. Under a realized tie-breaking $\succ$, student $i$ observes her rank $r_{is}=|\{j\in I: j\succ_s i\}|+1$ at each school $s$.

\begin{definition}\label{def:safe-weak}
Given an observation structure $\Gamma$, a weak priority profile $\succsim$, and a tie-breaking $\succ$ of $\succsim$, a matching $\mu$ is a \textbf{safe deviation} of $\varphi$ under $[P,\succsim,\succ]$ if the following conditions hold:
\begin{itemize}
  \item[(i)] $\varphi[P,\succ]\neq\mu$.
  \item[(ii)] For any $i\in I$, there exist $P' \in \mathcal{P}$ and a tie-breaking $\succ'$ of $\succsim$ such that $o_i(\varphi[P_i,P'_{-i},\succ'],\Gamma)=o_i(\mu,\Gamma)$ and $|\{j\in I: j\succ'_s i\}|+1=r_{is}$ for any $s\in S$.
  \item[(iii)] $\mu$ is not Pareto dominated by $\varphi[P,\succ]$.
\end{itemize}
\end{definition}

Compared to \Cref{def:safe}, \Cref{def:safe-weak} changes only condition~(ii). Under weak priorities, a student may explain her observation using a different tie-breaking of the same weak priority profile, provided that her observed rank at each school remains unchanged.

The unobserved part of the tie-breaking weakens the role of stability. Even under full observation, a student may not know whether the students assigned to her preferred schools have higher realized priority. Thus, $\DA$ is not credible and is not even more credible than $\TTC$ or $\EA$ when students observe the entire assignment. That is, part~(ii) of \Cref{thm:DA-stable} and part~(i) of \Cref{thm:DA-TTC} fail under weak priorities. However, the next result shows that the other rankings continue to hold.

\begin{prop}\label{prop:weak}
Fix any observation structure $\Gamma$. Under weak priorities,
\begin{itemize}
  \item[(i)] $\DA$ is more credible than any other stable mechanism.
  \item[(ii)] $\DA$ is strictly more credible than $\BM$.
  \item[(iii)] If $\Gamma(i)=\{i\}$ for any $i\in I$, then $\TTC$ and $\EA$ are strictly more credible than $\DA$.
\end{itemize}
\end{prop}

\Cref{prop:weak} is not a direct consequence of the strict-priority results. This is because, under strict priorities, all students' explanations must use the same priority profile, whereas under weak priorities, each student's explanation may use a different tie-breaking. The proof handles this difference by applying the relevant strict-priority argument separately to each student. Specifically, fix a student $i$ and a tie-breaking $\succ'$ used in her explanation. Holding $\succ'$ fixed, the strict-priority argument constructs an explanation under the other mechanism that gives $i$ the same observation. Thus, the same safe-deviation comparison extends to weak priorities.

\Cref{prop:weak} shows that the credibility rankings do not require publicly known strict priorities. It does not, however, rank $\DA$ against $\TTC$ or $\EA$ under an arbitrary observation structure. To obtain such a ranking, we next suppose that all schools are indifferent among all students.

\begin{prop}\label{prop:weak-indifferent-DA-TTC}
Suppose that all schools are indifferent among all students under the weak priority profile. Fix any observation structure $\Gamma$.
\begin{itemize}
  \item[(i)] $\TTC$ is strictly more credible than $\DA$.
  \item[(ii)] $\EA$ is strictly more credible than $\DA$.
\end{itemize}
\end{prop}

The intuition for \Cref{prop:weak-indifferent-DA-TTC} is similar to that for \Cref{thm:DA-TTC}. Under strict priorities, students can compare observed assignments with the priority order and rule out deviations that are inconsistent with stability. When all students are tied before the lottery, the weak priority profile provides no such comparison. As a result, condition~(ii) of \Cref{def:safe-weak} rules out fewer deviations under $\DA$, while condition~(iii) rules out more deviations under the efficient mechanisms $\TTC$ and $\EA$. Therefore, both mechanisms are more credible than $\DA$.

\Cref{prop:weak-indifferent-DA-TTC} shows that observing the entire assignment is not enough to make $\DA$ credible when the lottery remains partly hidden. This is because detecting deviations through stability requires the realized tie-breaking, from which each student observes only her own rank at each school. Therefore, a designer who relies on $\DA$ for credibility should disclose enough information about the realized tie-breaking for students to verify their assignments.

To relate our results to the literature on tie-breaking, we note that the setting above corresponds to multiple tie-breaking, under which schools may use different lotteries to break ties \citep{abdulkadirouglu2009strategy,ashlagi2020matters}. The literature also considers single tie-breaking, under which all schools use the same lottery. 
We next ask whether our results extend to this setting. 
We assume that students know that all schools use the same lottery, so each innocent explanation must also use the same lottery at every school.
Even with this restriction, \Cref{prop:weak} continues to hold. The proof requires only a small change: for each student, fix the common lottery used in her explanation and apply the relevant strict-priority argument to the resulting priority profile.
By contrast, \Cref{prop:weak-indifferent-DA-TTC} does not extend to single tie-breaking. When all schools are indifferent among students, the common lottery yields a common strict priority profile. In this case, $\DA$, $\TTC$, and $\EA$ coincide with serial dictatorship and are therefore equally credible.

\subsection{Affirmative action}\label{sec:affirmative}

Admission rules often differ across types of students. Cambridge's controlled choice policy seeks socioeconomic balance across schools, Chicago allocates 70\% of selective-enrollment seats equally across four neighborhood tiers, and New York City's screened high school programs reserve seats for students with disabilities.\footnote{For more on Cambridge public school admissions, see \url{https://www.cpsd.us/administration/student-registration-center/making-your-choices/about-controlled-choice}, accessed August 25, 2026. For Chicago, see \url{https://www.cps.edu/gocps/high-school/results/selective-enrollment/}, accessed August 25, 2026. For New York City, see \url{https://www.schools.nyc.gov/enrollment/enroll-grade-by-grade/high-school/screened-admissions}, accessed August 25, 2026.} 
Type-based admissions extend beyond socioeconomic status and disability.
Chinese colleges use separate admissions tracks for applicants to arts and sports programs, and U.S. colleges use a separate recruitment process for student-athletes.\footnote{For more on admissions to arts and sports programs in China, see the Beijing policy at \url{https://jw.beijing.gov.cn/bsfw/kdx/202505/t20250507_4083518.html}, accessed August 25, 2026. For more on U.S. student-athlete admissions, see \url{https://www.ncaa.org/eligibility-center/recruiting/}, accessed August 25, 2026.}
These examples motivate us to compare the credibility of different policies that favor designated types.

We compare two common policies: minority reserves and majority quotas. Minority reserves $r=(r_s)_{s \in S}$ give minority students priority for $r_s$ seats at school $s$ but allow majority students to use unfilled reserved seats. Majority quotas $q^M = (q_s^M)_{s \in S}$ limit the number of majority students assigned to school $s$ to $q_s^M$.

Let $I_m\subseteq I$ be the set of minority students. Let $\DA^r$ denote student-proposing deferred acceptance with minority reserves $r$, and let $\DA^q$ denote student-proposing deferred acceptance with majority quotas $q^M$. 
Each student knows the set of minority students $I_m$ and the policies $r$ and $q^M$, in addition to her own preference, the priority profile, and the quota profile.

We first suppose that students observe the entire assignment. The next result shows that both mechanisms are credible.

\begin{prop}\label{prop:affirmative-full-observation}
If $\Gamma(i)=I$ for any $i\in I$, then $\DA^r$ and $\DA^q$ are credible.
\end{prop}

The intuition for \Cref{prop:affirmative-full-observation} is as follows. Under full observation, condition~(ii) of \Cref{def:safe} requires any candidate deviation to be stable under the relevant policy. Moreover, for either policy, deferred acceptance selects the student-optimal matching among those satisfying this stability requirement \citep{hafalir2013effective}, so its outcome Pareto dominates any distinct candidate. Therefore, neither mechanism has a safe deviation.

Since both mechanisms are credible under full observation, \Cref{prop:affirmative-full-observation} does not rank the two policies. We therefore turn to a coarser observation structure in which a student observes her own assignment and the assignments of any subset of minority students, but not those of the other majority students.

\begin{prop}\label{prop:affirmative-reserve-quota}
Suppose $r_s+q_s^M=q_s$ and $q_s^M>0$ for any $s\in S$. Fix any observation structure $\Gamma$. If $\Gamma(i)\subseteq I_m\cup\{i\}$ for any $i\in I$, then $\DA^r$ is strictly more credible than $\DA^q$.
\end{prop}

The intuition for \Cref{prop:affirmative-reserve-quota} has two parts.
First, when some majority students' assignments are unobserved, each student's observation of a $\DA^r$ outcome can be reproduced by an explanation under $\DA^q$. This is possible because minority students' assignments remain feasible under the corresponding majority quotas, and any necessary changes affect only unobserved majority students.
Second, in this setting, $\DA^r$ weakly Pareto dominates $\DA^q$ \citep{hafalir2013effective}. This is because minority reserves allow majority students to use unclaimed reserved seats, whereas majority quotas may leave those seats empty. This welfare advantage makes condition~(iii) of \Cref{def:safe} more restrictive under $\DA^r$. 
Together, these arguments explain the credibility ranking.

\Cref{prop:affirmative-reserve-quota} shows that disclosing minority students' assignments is not enough to make $\DA^q$ credible, a limitation that also appears in practice.
Before 2024, some Chinese universities admitted students with artistic talent under policies similar to majority quotas.\footnote{For more on admissions for students with artistic talent in China, see the 2021 policy of Tsinghua University at \url{https://join-tsinghua.edu.cn/info/1033/1364.htm}, accessed August 25, 2026.} 
Although universities were required to disclose qualified candidates and admission requirements, this policy still raised concerns about opaque selection and favoritism.\footnote{For more on this requirement, see \url{http://www.moe.gov.cn/srcsite/A15/moe_776/tslxzs/202211/t20221111_984077.html}, accessed August 25, 2026. For more on these concerns, see Jia Song, Ying Xie, and Xiyuan Chen, \textit{Multiple Universities End High-Level Art-Troupe Admissions Ahead of Schedule, Narrowing a Shortcut to Prestigious Universities}, Xinhua News Agency, March 1, 2022, \url{https://www.news.cn/politics/2022-03/01/c_1128426776.htm}, accessed August 25, 2026.} 
These concerns, among other reasons, led China to end this admission policy in 2024.\footnote{For more on the end of this policy, see \url{http://www.moe.gov.cn/srcsite/A15/moe_776/s3109/202109/t20210923_566071.html}, accessed August 25, 2026.} In our framework, these concerns reflect a credibility problem: even with the minority students' assignments disclosed, $\DA^q$ is not credible, and it is strictly less credible than $\DA^r$.

\section{Conclusion}\label{sec:conclusion}

The success of a centralized school choice system depends not only on the desirability of its assignments but also on whether families trust the designer to follow the announced rule. 
This trust may be difficult to maintain if families cannot rule out that the designer has deviated from the rule to favor some students. 
Just as mechanisms are compared by other desirable properties, we compare widely used school choice mechanisms by their credibility, namely their robustness to such deviations. 
How credible a mechanism is, however, depends strongly on how much of the assignment is disclosed. 
For this reason, allowing families to observe only part of the realized assignment is a central feature of our analysis.

Our main conclusion is that credibility depends on both the mechanism and what families observe about the realized assignment.
On the one hand, deferred acceptance is more credible than the Boston mechanism and any other stable mechanism, regardless of how much of the assignment families observe. 
On the other hand, its ranking relative to top trading cycles and efficiency-adjusted deferred acceptance reverses with the amount of disclosure: deferred acceptance is more credible when families observe the entire assignment, whereas the two efficient mechanisms are more credible when families observe only their own assignment. 
This suggests that deferred acceptance is the most credible choice in systems where the designer is willing to disclose the entire assignment. 
When such disclosure is not feasible, however, alternatives such as top trading cycles or efficiency-adjusted deferred acceptance may be preferable.


\appendix
\section{Omitted proofs}\label[appendix]{sec:proof}

Given $\mu\in\mathcal{M}$, let $P^\mu$ be the preference profile under which, for any student $i$, $\mu(i)$ is her only acceptable school if $\mu(i)\in S$, and no school is acceptable if $\mu(i)=\emptyset$.
Let $R_i^\mu$ denote the weak preference relation induced by $P_i^\mu$.

\subsection{Proof of Section~\ref{sec:main}}

\subsubsection{Proof of Theorem~\ref{thm:DA-stable}}

\begin{lm}\label{lm:DA-stable}
Fix a strict priority profile $\succ$, a preference profile $P$, a student $i\in I$, and a matching $\mu$ that is stable under $P$ and $\succ$. For any stable mechanism $\varphi$, $\varphi[P_i,P^\mu_{-i},\succ]=\mu$.
\end{lm}
  
\begin{proof}[Proof of \Cref{lm:DA-stable}]
Let $\eta=\varphi[P_i,P^\mu_{-i},\succ]$. The stability of $\varphi$ ensures that $\eta$ is stable under $(P_i,P^\mu_{-i})$.
Since $\mu$ is stable under $P$, it is also stable under $(P_i,P^\mu_{-i})$.
By the Rural Hospitals Theorem~\citep{gale1985some,roth1986allocation}, the set of matched students is the same in $\mu$ and $\eta$, and $|\mu(s)|=|\eta(s)|$ at every school $s\in S$.

Fix a student $j\neq i$. If $\mu(j)=\emptyset$, then $\eta(j)=\emptyset$. If $\mu(j)\in S$, then $j$ is matched under $\eta$, and individual rationality together with the definition of $P^\mu_j$ gives $\eta(j)=\mu(j)$. Moreover, equality of enrollment at every school implies $\eta(i)=\mu(i)$. Therefore, $\varphi[P_i,P^\mu_{-i},\succ]=\eta=\mu$.
\end{proof}

\begin{proof}[Proof of \Cref{thm:DA-stable}]
We start by proving part~(i). Fix a stable mechanism $\varphi$, and let $\mu$ be a safe deviation of $\DA$ under $P \in \mathcal{P}$. We show that $\mu$ is also a safe deviation of $\varphi$. Fix a student $i\in I$. By condition~(ii) of \Cref{def:safe}, there exists $P' \in \mathcal{P}$ such that $o_i(\DA[P_i,P'_{-i}],\Gamma)=o_i(\mu,\Gamma)$. Let $\eta=\DA[P_i,P'_{-i}]$. Since $\eta$ is stable under $(P_i,P'_{-i})$, \Cref{lm:DA-stable} gives $\varphi[P_i,P^\eta_{-i}]=\eta$. Thus, $o_i(\varphi[P_i,P^\eta_{-i}],\Gamma)=o_i(\eta,\Gamma)=o_i(\mu,\Gamma)$.

It remains to verify conditions~(i) and~(iii) of \Cref{def:safe}.
Suppose first that $\mu=\varphi[P]$. Then, the student optimality of $\DA$ and $\mu\neq\DA[P]$ imply that $\DA[P]$ Pareto dominates $\varphi[P]=\mu$, contradicting condition~(iii) of \Cref{def:safe}. Thus, $\mu\neq\varphi[P]$.
Suppose second that $\mu$ is Pareto dominated by $\varphi[P]$. Then, the student optimality of $\DA$ implies that $\DA[P]$ Pareto dominates $\mu$, again contradicting condition~(iii) of \Cref{def:safe}. 
Therefore, $\mu$ is a safe deviation of $\varphi$ and $\DA$ is more credible than $\varphi$.

To prove part~(ii), we first show that if $\Gamma(i) = I$ for any $i \in I$, $\DA$ is credible. Suppose, toward a contradiction, that $\mu$ is a safe deviation of $\DA$ under some $P\in \mathcal{P}$. By Theorem~3 of \citet{moller2026transparent}, $\mu$ is stable under $P$. Since $\DA$ is the student-optimal stable mechanism and $\mu\neq \DA[P]$, $\DA[P]$ Pareto dominates $\mu$. This contradicts condition~(iii) of \Cref{def:safe}. Therefore, $\DA$ is credible when students observe the entire assignment.

On the other hand, suppose, toward a contradiction, that there exists a credible stable mechanism $\varphi \neq \DA$. Fix $P \in \mathcal{P}$ such that $\varphi[P] \neq \DA[P]$, and let $\eta = \DA[P]$. Since $\eta$ is the student-optimal stable matching, $\eta$ is not Pareto dominated by $\varphi[P]$.
Fix a student $i\in I$. Applying \Cref{lm:DA-stable} to the matching $\eta$ gives $\varphi[P_i,P^\eta_{-i}]=\eta$ and thus $o_i(\varphi[P_i,P^\eta_{-i}],\Gamma)=o_i(\eta,\Gamma)$.
Therefore, $\eta$ is a safe deviation of $\varphi$ under $P$, which contradicts the fact that $\varphi$ is credible.
\end{proof}

\subsubsection{Proof of Theorem~\ref{thm:DA-BM}}

\begin{lm}\label{lm:DA-BM}
Fix a strict priority profile $\succ$, a preference profile $P$, a student $i\in I$, and a matching $\mu$ that is stable under $P$ and $\succ$. Then $\BM[P_i,P^\mu_{-i},\succ]=\mu$.
\end{lm}

\begin{proof}[Proof of \Cref{lm:DA-BM}]
Let $s=\BM[P_i,P^\mu_{-i},\succ](i)$.
Suppose first that $s\mathrel{P_i}\mu(i)$. The stability of $\mu$ under $P$ implies that $s$ is full under $\mu$ and any student in $\mu(s)$ has higher priority than $i$ at $s$. Since any student in $\mu(s)$ applies to $s$ in the first round under $(P_i,P^\mu_{-i})$, $\BM$ cannot assign $s$ to $i$, a contradiction. Thus, $\mu(i)\mathrel{R_i}s$.

Suppose next that $\mu(i)\mathrel{P_i}s$. For any student $j\neq i$, the definition of $P^\mu_j$ and individual rationality imply that $\mu(j)\mathrel{R^\mu_j}\BM[P_i,P^\mu_{-i},\succ](j)$. Thus, $\mu$ Pareto dominates $\BM[P_i,P^\mu_{-i},\succ]$ under $(P_i,P^\mu_{-i})$, contradicting the Pareto efficiency of $\BM$. Therefore, $\BM[P_i,P^\mu_{-i},\succ](i)=\mu(i)$. 

Since $\BM[P^\mu,\succ]=\mu$ by construction, we have $\BM[P_i,P^\mu_{-i},\succ](i)=\BM[P^\mu,\succ](i)$.
Since $\BM$ is nonbossy \citep{harless2014school}, changing only $i$'s report from $P^\mu_i$ to $P_i$ without changing her assignment cannot change the matching. Thus, $\BM[P_i,P^\mu_{-i},\succ]=\BM[P^\mu,\succ]=\mu$.
\end{proof}

\begin{lm}\label{lm:BM-unstable}
Fix a strict priority profile $\succ$ and a preference profile $P$. If $\BM[P,\succ]$ is unstable, then $\BM$ has a safe deviation under $[P,\succ]$ for any observation structure.
\end{lm}

\begin{proof}[Proof of \Cref{lm:BM-unstable}]
Let $\mu=\BM[P,\succ]$. Since $\BM$ is individually rational and nonwasteful, the instability of $\BM[P,\succ]$ implies that there exist a student $i^*$ and a school $s$ such that $s\mathrel{P_{i^*}}\mu(i^*)$ and $i^*\succ_s j$ for some $j\in\mu(s)$.

First, we show that student $i^*$ can obtain $s$ by reporting it as her only acceptable school. To see this, suppose that at least $q_s$ students with higher priority than $i^*$ apply to $s$ in the first round. Since these students fill all seats at $s$ in that round, the lower-priority student $j$ is not assigned to $s$, a contradiction. Thus, fewer than $q_s$ higher-priority students apply to $s$ in the first round, and $i^*$ is accepted when she reports $s$ as her only acceptable school.

Among the schools that $i^*$ can obtain by changing her report while holding $P_{-i^*}$ fixed, let $s^*$ be her most-preferred one, and let $P^*_{i^*}$ be a report that assigns her to $s^*$.
That is, for any $P'_{i^*}$, $\BM[P^*_{i^*},P_{-i^*},\succ](i^*) \mathrel{R_{i^*}} \BM[P'_{i^*},P_{-i^*},\succ](i^*)$.
Let $\eta=\BM[P^*_{i^*},P_{-i^*},\succ]$. 
The preceding paragraph implies that $s^*\mathrel{P_{i^*}}\mu(i^*)$.

We claim that $\BM[P_{i^*},P^\eta_{-i^*},\succ]=\eta$. 
Fix a school $t\mathrel{P_{i^*}}s^*$. We first show that at least $q_t$ students with higher priority than $i^*$ must apply to $t$ in the first round under $P_{-i^*}$. Suppose, toward a contradiction, that fewer than $q_t$ such students apply. The argument used for $s$ then shows that $i^*$ obtains $t$ by reporting it as her only acceptable school, contradicting the choice of $s^*$. Thus, under $(P^*_{i^*},P_{-i^*})$, school $t$ is filled in the first round with students who have higher priority than $i^*$.
Since any student other than $i^*$ applies to her assignment under $\eta$ in the first round under $(P_{i^*},P^\eta_{-i^*})$, $i^*$ is rejected by any school that she prefers to $s^*$. Then, since $\eta(i^*)=s^*$, at most $q_{s^*}-1$ other students apply to $s^*$, so a seat remains for $i^*$ when she applies. Moreover, any other student is accepted at her assignment under $\eta$. Therefore, $\BM[P_{i^*},P^\eta_{-i^*},\succ]=\eta$.

Finally, we show that $\eta$ is a safe deviation of $\BM$ under any observation structure $\Gamma$.
For any student $i\neq i^*$, $o_i(\BM[P^*_{i^*},P_{-i^*},\succ], \Gamma) = o_i(\eta, \Gamma)$.
For $i^*$, $o_{i^*}(\BM[P_{i^*},P^\eta_{-i^*},\succ], \Gamma) = o_{i^*}(\eta, \Gamma)$.
Moreover, $\eta\neq\mu$, and $\mu$ does not Pareto dominate $\eta$ since $i^*$ strictly prefers $\eta(i^*)$ to $\mu(i^*)$. Therefore, $\eta$ is a safe deviation of $\BM$.
\end{proof}

\begin{proof}[Proof of \Cref{thm:DA-BM}]
Let $\mu$ be a safe deviation of $\DA$ under $P\in \mathcal{P}$. We consider two cases. 

Suppose first that $\BM[P]$ is stable. Since $\DA[P]$ is the student-optimal stable matching and $\BM[P]$ is Pareto efficient, $\BM[P]=\DA[P]$. Thus, $\mu\neq\BM[P]$, and $\mu$ is not Pareto dominated by $\BM[P]$.
Fix a student $i\in I$. Since $\mu$ is a safe deviation of $\DA$ under $P$, there exists $P'_{-i}$ such that $o_i(\DA[P_i,P'_{-i}],\Gamma)=o_i(\mu,\Gamma)$. Let $\eta=\DA[P_i,P'_{-i}]$. Since $\eta$ is stable under $(P_i,P'_{-i})$, \Cref{lm:DA-BM} implies that $\BM[P_i,P^\eta_{-i}]=\eta$. Thus, $o_i(\BM[P_i,P^\eta_{-i}],\Gamma) = o_i(\eta,\Gamma) = o_i(\mu,\Gamma)$. Therefore, $\mu$ is a safe deviation of $\BM$ under $P$.

Suppose second that $\BM[P]$ is not stable. By \Cref{lm:BM-unstable}, $\BM$ has a safe deviation under $P$. Therefore, $\DA$ is more credible than $\BM$.
  
To show that $\DA$ is strictly more credible than $\BM$, consider the following example. Consider three students $\{i_1,i_2,i_3\}$ and two schools $\{s_1,s_2\}$. Let $q_{s_1}=q_{s_2}=1$. Consider the following preference profile and priority profile:
\begin{table}[H]
\centering
\begin{tabular}{lll|ll}
  $P_{i_1}$ & $P_{i_2}$ & $P_{i_3}$ & $\succ_{s_1}$ & $\succ_{s_2}$ \\ \hline
  $s_1$ & $s_1$ & $s_2$ & $i_1$ & $i_1$ \\
  $s_2$ & $s_2$ & $s_1$ & $i_2$ & $i_2$ \\
  $\emptyset$ & $\emptyset$ & $\emptyset$ & $i_3$ & $i_3$
\end{tabular}
\end{table}
Consider the matchings
\[
\mu=\begin{pmatrix}
i_1 & i_2 & i_3 \\
s_1 & s_2 & \emptyset
\end{pmatrix},
\qquad
\eta=\begin{pmatrix}
i_1 & i_2 & i_3 \\
s_1 & \emptyset & s_2
\end{pmatrix}.
\]
Notice that $\DA[P]=\mu$ and $\BM[P]=\eta$.

First, we show that $\DA$ has no safe deviation under $P$ for any observation structure. Notice that it is sufficient to consider $\Gamma(i)=\{i\}$ for any $i\in I$. Suppose, toward a contradiction, that $\nu$ is a safe deviation. For any $P'_{-i_1}$, since $i_1$ ranks $s_1$ first and has the highest priority there, stability implies that $\nu(i_1)=\DA[P_{i_1},P'_{-i_1}](i_1)=s_1$. Moreover, for any $P'_{-i_2}$, since $i_2$ is ranked second by both schools, stability implies that $\nu(i_2)=\DA[P_{i_2},P'_{-i_2}](i_2)\in\{s_1,s_2\}$.
Then, feasibility gives $\nu(i_2)=s_2$ and $\nu(i_3)=\emptyset$, so $\nu=\mu=\DA[P]$, contradicting condition~(i) of \Cref{def:safe}.

Second, we show that $\mu$ is a safe deviation of $\BM$ under $P$ for any observation structure. Since $\mu$ is stable under $(P,\succ)$, \Cref{lm:DA-BM} gives $\BM[P_i,P^\mu_{-i}]=\mu$ for any $i\in I$. Thus, $o_i(\BM[P_i,P^\mu_{-i}],\Gamma)=o_i(\mu,\Gamma)$ for any $i\in I$. Moreover, $\mu\neq\BM[P]=\eta$, and $\eta$ does not Pareto dominate $\mu$. Therefore, $\mu$ is a safe deviation of $\BM$ under $P$.
\end{proof}

\subsubsection{Proof of Theorem~\ref{thm:DA-TTC}}

\begin{lm}\label{lm:DA-TTC}
Fix a strict priority profile $\succ$, a preference profile $P$, and a student $i\in I$. For any $\varphi\in\{\TTC,\EA\}$, there exists $\widehat P_{-i}$ such that $\DA[P_i, \widehat P_{-i},\succ](i)=\varphi[P,\succ](i)$.
\end{lm}

\begin{proof}[Proof of \Cref{lm:DA-TTC}]
Theorem~2 and Proposition~2 of \citet{decerf2024incontestable} imply that $\EA$, $\TTC$, and $\DA$ are $i$-indistinguishable. That is, there exists $P'_{-i}$ such that $\DA[P_i, P'_{-i},\succ](i)=\EA[P,\succ](i)$, and there exists $P''_{-i}$ such that $\DA[P_i, P''_{-i},\succ](i)=\TTC[P,\succ](i)$.
\end{proof}

\begin{proof}[Proof of \Cref{thm:DA-TTC}]
To prove part~(i), first notice that when $\Gamma(i)=I$ for any $i\in I$, $\DA$ is credible by \Cref{thm:DA-stable}. Thus, $\DA$ is more credible than $\TTC$ and $\EA$. 

To show that both comparisons are strict, it remains to show that neither $\TTC$ nor $\EA$ is credible under full observation. Consider three students $\{i_1,i_2,i_3\}$ and two schools $\{s_1,s_2\}$. Let $q_{s_1}=q_{s_2}=1$. Consider the following preference profile and priority profile:
\begin{table}[H]
\centering
\begin{tabular}{lll|ll}
  $P_{i_1}$ & $P_{i_2}$ & $P_{i_3}$ & $\succ_{s_1}$ & $\succ_{s_2}$ \\ \hline
  $s_2$ & $s_1$ & $s_2$ & $i_1$ & $i_2$ \\
  $s_1$ & $\emptyset$ & $\emptyset$ & $i_2$ & $i_3$ \\
  $\emptyset$ &  & & $i_3$ & $i_1$
\end{tabular}
\end{table}
Consider the matchings
\[
\mu=\begin{pmatrix}
i_1 & i_2 & i_3 \\
s_1 & \emptyset & s_2
\end{pmatrix},
\qquad
\eta=\begin{pmatrix}
i_1 & i_2 & i_3 \\
s_2 & s_1 & \emptyset
\end{pmatrix}.
\]
Notice that $\DA[P]=\EA[P]=\mu$ and $\TTC[P]=\eta$.

First, we show that $\mu$ is a safe deviation of $\TTC$. Indeed, $\mu\neq\TTC[P]=\eta$, and $\mu$ is not Pareto dominated by $\eta$. Moreover, $\TTC[P_i,P^\mu_{-i}]=\mu$ for any $i\in I$, so $o_i(\TTC[P_i,P^\mu_{-i}],\Gamma)=o_i(\mu,\Gamma)$. Therefore, $\mu$ is a safe deviation of $\TTC$, and $\DA$ is strictly more credible than $\TTC$ when students observe the entire assignment. 

Second, we show that $\eta$ is a safe deviation of $\EA$. 
For student $i_1$, let both $i_2$ and $i_3$ report $s_1$, $s_2$, and $\emptyset$ in that order. Then $\EA[P_{i_1},P'_{-i_1}]=\eta$. 
For student $i_2$, let $i_1$ report $P_{i_1}$ and let $i_3$ report $s_1$ as her only acceptable school. Then $\EA[P_{i_2},P'_{-i_2}]=\eta$. 
For student $i_3$, let $i_1$ report $P_{i_1}$ and let $i_2$ report $s_1$, $s_2$, and $\emptyset$ in that order. Then $\EA[P_{i_3},P'_{-i_3}]=\eta$. 
Thus, every student has an innocent explanation for $\eta$. Moreover, $\eta\neq\EA[P]=\mu$, and $\mu$ does not Pareto dominate $\eta$. Therefore, $\eta$ is a safe deviation of $\EA$, and $\DA$ is strictly more credible than $\EA$ when students observe the entire assignment. 

We continue with part~(ii). 
Suppose that $\Gamma(i)=\{i\}$ for any $i\in I$. We first compare $\TTC$ and $\DA$. Fix a preference profile $P \in \mathcal{P}$ and a safe deviation $\mu$ of $\TTC$ under $P$. We consider two cases. 

Suppose first that $\TTC[P]=\DA[P]$. We claim that $\mu$ is a safe deviation of $\DA$ under $P$. 
Since $\TTC[P]=\DA[P]$, we have $\mu\neq \DA[P]$ and $\mu$ is not Pareto dominated by $\DA[P]$. 
Fix a student $i\in I$. By condition~(ii) of \Cref{def:safe}, there exists $P'\in\mathcal{P}$ such that $o_i(\TTC[P_i,P'_{-i}],\Gamma)=o_i(\mu,\Gamma)$. Combining with $\Gamma(i)=\{i\}$, we have $\TTC[P_i,P'_{-i}](i)=\mu(i)$. Then, \Cref{lm:DA-TTC} implies that there exists $P''_{-i}$ such that $\DA[P_i,P''_{-i}](i)=\TTC[P_i,P'_{-i}](i)=\mu(i)$. Thus, $o_i(\DA[P_i,P''_{-i}],\Gamma)=o_i(\mu,\Gamma)$, and $\mu$ is a safe deviation of $\DA$ under $P$.

Suppose second that $\TTC[P]\neq\DA[P]$. We claim that $\TTC[P]$ is a safe deviation of $\DA$ under $P$. Indeed, the Pareto efficiency of $\TTC$ ensures that $\TTC[P]$ is not Pareto dominated by $\DA[P]$. Moreover, for any student $i\in I$, \Cref{lm:DA-TTC} implies that there exists $P'_{-i}$ such that $\DA[P_i,P'_{-i}](i)=\TTC[P](i)$. Since $\Gamma(i)=\{i\}$, we have $o_i(\DA[P_i,P'_{-i}],\Gamma)=o_i(\TTC[P],\Gamma)$. Thus, $\TTC[P]$ is a safe deviation of $\DA$ under $P$, and $\TTC$ is more credible than $\DA$ when students observe only their own assignment.

We next compare $\EA$ and $\DA$. Fix a preference profile $P \in \mathcal{P}$ and a safe deviation $\mu$ of $\EA$ under $P$.
We show that $\mu$ is also a safe deviation of $\DA$ under $P$. 

First, we show that $\DA[P]\neq \mu$. Otherwise, since $\EA$ weakly Pareto dominates $\DA$ and $\EA[P]\neq \DA[P]$, $\EA[P]$ Pareto dominates $\DA[P]=\mu$, contradicting condition~(iii) of \Cref{def:safe}. Thus, $\DA[P]\neq \mu$.

Next, we show that $\mu$ is not Pareto dominated by $\DA[P]$. Otherwise, since $\EA$ weakly Pareto dominates $\DA$, $\EA[P]$ Pareto dominates $\mu$, contradicting condition~(iii) of \Cref{def:safe}. Thus, $\mu$ is not Pareto dominated by $\DA[P]$. 

Finally, by condition~(ii) of \Cref{def:safe}, for any $i\in I$ there exists $P'\in\mathcal{P}$ such that $o_i(\EA[P_i,P'_{-i}],\Gamma)=o_i(\mu,\Gamma)$. Since $\Gamma(i)=\{i\}$, this implies that $\EA[P_i,P'_{-i}](i)=\mu(i)$. Then, \Cref{lm:DA-TTC} implies that there exists $P''_{-i}$ such that $\DA[P_i,P''_{-i}](i)=\EA[P_i,P'_{-i}](i)=\mu(i)$. Therefore, $o_i(\DA[P_i,P''_{-i}],\Gamma)=o_i(\mu,\Gamma)$. We conclude that $\mu$ is a safe deviation of $\DA$ under $P$. Thus, $\EA$ is more credible than $\DA$ when students observe only their own assignment.

We have shown that $\TTC$ and $\EA$ are more credible than $\DA$ when students observe only their own assignment. To show that both comparisons are strict, consider three students $\{i_1,i_2,i_3\}$ and two schools $\{s_1,s_2\}$. Let $q_{s_1}=q_{s_2}=1$. Consider the following preference profile and priority profile:\footnote{Notice that this example differs from the full-observation example above only in student $i_2$'s preference.}
\begin{table}[H]
\centering
\begin{tabular}{lll|ll}
  $P_{i_1}$ & $P_{i_2}$ & $P_{i_3}$ & $\succ_{s_1}$ & $\succ_{s_2}$ \\ \hline
  $s_2$ & $s_1$ & $s_2$ & $i_1$ & $i_2$ \\
  $s_1$ & $s_2$ & $\emptyset$ & $i_2$ & $i_3$ \\
  $\emptyset$ & $\emptyset$ & & $i_3$ & $i_1$
\end{tabular}
\end{table}
Consider the matchings
\[
\mu=\begin{pmatrix}
i_1 & i_2 & i_3 \\
s_1 & s_2 & \emptyset
\end{pmatrix},
\qquad
\eta=\begin{pmatrix}
i_1 & i_2 & i_3 \\
s_2 & s_1 & \emptyset
\end{pmatrix}.
\]
Notice that $\DA[P]=\mu$ and $\EA[P]=\TTC[P]=\eta$.

First, we show that $\varphi\in\{\EA,\TTC\}$ has no safe deviation under $P$. Suppose, toward a contradiction, that $\varphi$ has a safe deviation $\nu$. Under $\TTC$, student $i_1$ cannot remain unmatched while she has top priority at the acceptable school $s_1$. Under $\EA$, stability of $\DA$ implies that $\DA$ assigns $i_1$ to $s_1$ or $s_2$, and the weak Pareto improvement from $\DA$ to $\EA$ implies that $\EA$ assigns $i_1$ to $s_1$ or $s_2$. Thus, $\nu(i_1)\in\{s_1,s_2\}$ for either mechanism. 
The analogous argument for $i_2$, who has top priority at $s_2$, gives $\nu(i_2)\in\{s_1,s_2\}$. Feasibility leaves only $\eta$ and $\mu$ as possibilities. Since $\eta=\varphi[P]$ Pareto dominates $\mu$, conditions~(i) and~(iii) rule out both possibilities, a contradiction.

Second, we show that $\eta$ is a safe deviation of $\DA$. Indeed, $\eta\neq\DA[P]$ and $\eta$ is not Pareto dominated by $\DA[P]$. For $i_1$ and $i_2$, $\DA[P_i,P^\eta_{-i}]=\eta$. For $i_3$, let $i_1$ report no school as acceptable and let $i_2$ report only $s_2$ as acceptable. Then $\DA[P_{i_3},P'_{-i_3}](i_3)=\eta(i_3)$. Thus, $\eta$ is a safe deviation of $\DA$.\footnote{This example also shows that any stable mechanism can admit a safe deviation when every student except one observes the entire assignment. Suppose that $i_1$ and $i_2$ observe the entire assignment and $i_3$ observes only her own assignment. The explanations above for $i_1$ and $i_2$ produce $\eta$, so $\eta$ remains a safe deviation of $\DA$. By \Cref{thm:DA-stable}~(i), every stable mechanism then admits a safe deviation under this problem.}
\end{proof}

\subsection{Proof of Section~\ref{sec:constrained}}

Fix an observation structure $\Gamma$. We say that a mechanism $\varphi$ admits \textbf{simple explanations} under $\Gamma$ if, for any $P\in \mathcal{P}$, any $i\in I$, and any safe deviation $\mu$ of $\varphi$ under $P$, there exists $\eta\in\mathcal{M}$ such that $o_i(\varphi[P_i,P^\eta_{-i}],\Gamma)=o_i(\mu,\Gamma)$. 

\subsubsection{Proof of Proposition~\ref{prop:constrained}}

\begin{lm}\label{lm:constrained}
Fix an observation structure $\Gamma$, a constraint $k\geq 1$, and two individually rational mechanisms $\varphi$ and $\psi$. Suppose that $\psi$ admits simple explanations under $\Gamma$.

\begin{itemize}
  \item[(i)] If $\varphi$ is credible, then $\varphi^k$ is credible. 
  \item[(ii)] If $\varphi$ is more credible than $\psi$, then $\varphi^k$ is more credible than $\psi^k$.  
\end{itemize}
\end{lm}

\begin{proof}[Proof of \Cref{lm:constrained}]
To prove part~(i), suppose, toward a contradiction, that $\mu$ is a safe deviation of $\varphi^k$ under $P$. Then, $\varphi[P^k]=\varphi^k[P]\neq\mu$. For any $i\in I$, choose $P'_{-i}$ such that $o_i(\varphi^k[P_i,P'_{-i}],\Gamma)=o_i(\mu,\Gamma)$. Since $\varphi^k[P_i,P'_{-i}]=\varphi[P_i^k,{P'_{-i}}^k]$, we have $o_i(\varphi[P_i^k,{P'_{-i}}^k],\Gamma)=o_i(\mu,\Gamma)$.

It remains to show that $\mu$ is not Pareto dominated by $\varphi[P^k]$ under $P^k$. Since $\mu\neq\varphi^k[P]$ and $\mu$ is not Pareto dominated by $\varphi^k[P]$ under $P$, there exists $j\in I$ such that $\mu(j)\mathrel{P_j}\varphi[P^k](j)$. Then, individual rationality of $\varphi$ gives $\varphi[P^k](j)\mathrel{R^k_j}\emptyset$. 
Moreover, since $j \in \Gamma(j)$, there exists $\widehat P_{-j}$ such that $\varphi^k[P_j, \widehat P_{-j}](j) = \mu(j)$, and thus $\varphi[P_j^k, \widehat P_{-j}^k](j) = \mu(j)$. Then, individual rationality of $\varphi$ gives $\mu(j)\mathrel{R^k_j}\emptyset$. Since both assignments lie in the part of $j$'s preference list whose ordering is unchanged by truncation, $\mu(j)\mathrel{P^k_j}\varphi[P^k](j)$. Thus, $\mu$ is not Pareto dominated by $\varphi[P^k]$ under $P^k$. Therefore, $\mu$ is a safe deviation of $\varphi$ under $P^k$, contradicting the credibility of $\varphi$.

We now prove part~(ii). Let $\mu$ be a safe deviation of $\varphi^k$ under $P$. The preceding argument shows that $\mu$ is a safe deviation of $\varphi$ under $P^k$. Since $\varphi$ is more credible than $\psi$, there exists a safe deviation $\eta$ of $\psi$ under $P^k$. We claim that $\eta$ is a safe deviation of $\psi^k$ under $P$. 

By condition~(i) of \Cref{def:safe}, $\psi^k[P]=\psi[P^k]\neq\eta$.
Moreover, since $\psi$ admits simple explanations under $\Gamma$, for any $i\in I$, there exists a matching $\nu$ such that $o_i(\psi[P_i^k,P^\nu_{-i}],\Gamma)=o_i(\eta,\Gamma)$. Then, since any report in $P^\nu_{-i}$ contains at most one acceptable school, truncation does not change it. Thus, $o_i(\psi^k[P_i,P^\nu_{-i}],\Gamma)=o_i(\psi[P_i^k,P^\nu_{-i}],\Gamma)=o_i(\eta,\Gamma)$.
Finally, since $\eta$ is not Pareto dominated by $\psi[P^k]$ under $P^k$, there exists $j\in I$ such that $\eta(j)\mathrel{P^k_j}\psi[P^k](j)$.
Thus, $\eta(j)\mathrel{P_j}\psi[P^k](j)=\psi^k[P](j)$ and $\eta$ is not Pareto dominated by $\psi^k[P]$ under $P$. Therefore, $\eta$ is a safe deviation of $\psi^k$ under $P$, and thus $\varphi^k$ is more credible than $\psi^k$.
\end{proof}

\begin{lm}\label{lm:simple}
Fix any observation structure $\Gamma$.
\begin{itemize}
  \item[(i)] Any stable mechanism admits simple explanations.
  \item[(ii)] $\BM$ admits simple explanations.
\end{itemize}
\end{lm}
  
\begin{proof}[Proof of \Cref{lm:simple}]
To prove part~(i), fix a stable mechanism $\varphi$, and let $\mu$ be a safe deviation of $\varphi$ under $P\in\mathcal{P}$. Fix a student $i\in I$. By condition~(ii) of \Cref{def:safe}, there exists $P'\in \mathcal{P}$ such that $o_i(\varphi[P_i,P'_{-i}],\Gamma)=o_i(\mu,\Gamma)$. Let $\eta=\varphi[P_i,P'_{-i}]$. Since $\eta$ is stable under $(P_i,P'_{-i})$, \Cref{lm:DA-stable} gives $\varphi[P_i,P^\eta_{-i}]=\eta$. Thus, $o_i(\varphi[P_i,P^\eta_{-i}],\Gamma)=o_i(\eta,\Gamma)=o_i(\mu,\Gamma)$, and therefore $\varphi$ admits simple explanations under $\Gamma$.

To prove part~(ii), let $\mu$ be a safe deviation of $\BM$ under $P\in\mathcal{P}$. Fix a student $i\in I$. By condition~(ii) of \Cref{def:safe}, there exists $P'\in \mathcal{P}$ such that $o_i(\BM[P_i,P'_{-i}],\Gamma)=o_i(\mu,\Gamma)$. Let $\eta=\BM[P_i,P'_{-i}]$. We next show that $\BM[P_i,P^\eta_{-i}]=\eta$.

Fix a student $j\neq i$. Consider a preference profile $\widehat P$ such that $\widehat P_i=P_i$, $\widehat P_j=P'_j$, $\widehat P_\ell\in\{P'_\ell,P^\eta_\ell\}$ for any $\ell\in I\setminus\{i,j\}$. Suppose that $\BM[\widehat P]=\eta$. Let $\nu=\BM[P^\eta_j,\widehat P_{-j}]$. If $\eta(j)=\emptyset$, individual rationality gives $\nu(j)=\emptyset=\eta(j)$. If $\eta(j)\in S$ and $j$ applies to $\eta(j)$ in the first round under $\widehat P$, changing her report to $P^\eta_j$ preserves that application and gives $\nu(j)=\eta(j)$. If $j$ applies to $\eta(j)$ in a later round under $\widehat P$, the school has a vacant seat after the first round. Thus, moving her application to the first round gives $\nu(j)=\eta(j)$. Therefore, $\nu(j)=\eta(j)$ in any case. Since $\BM$ is nonbossy \citep{harless2014school}, $\nu=\BM[P^\eta_j,\widehat P_{-j}]=\BM[\widehat P]=\eta$.

Apply the argument above successively to each student in $I\setminus\{i\}$. At any step, the outcome remains $\eta$, so the argument applies to the next student's change from her report in $P'$ to her report in $P^\eta$. Thus, $\BM[P_i,P^\eta_{-i}]=\eta$, and $o_i(\BM[P_i,P^\eta_{-i}],\Gamma)=o_i(\eta,\Gamma)=o_i(\mu,\Gamma)$. Therefore, $\BM$ admits simple explanations under $\Gamma$.
\end{proof}

\begin{proof}[Proof of \Cref{prop:constrained}]
For part~(i), the result follows from \Cref{thm:DA-stable}~(i), \Cref{lm:constrained}~(ii), and \Cref{lm:simple}~(i).

For part~(ii), the result follows from \Cref{thm:DA-BM}, \Cref{lm:constrained}~(ii), and \Cref{lm:simple}~(ii).
The strict relation comes from the example in the proof of \Cref{thm:DA-BM}.

For part~(iii), the credibility of $\DA^k$ follows from \Cref{thm:DA-stable}~(ii) and \Cref{lm:constrained}~(i).
The strict relations follow from the examples in the proof of \Cref{thm:DA-TTC}.

We next show that no other constrained stable mechanism is credible. Let $\varphi^k$ be a constrained stable mechanism such that $\varphi^k[P]\neq\DA^k[P]$ for some $P$, and let $\mu=\DA^k[P]=\DA[P^k]$. Since $\mu$ is the student-optimal stable matching under $P^k$ and $\varphi^k[P]\neq\mu$, $\mu$ Pareto dominates $\varphi^k[P]$ under $P^k$. Since $\DA^k$ and $\varphi^k$ are individually rational, $\mu$ Pareto dominates $\varphi^k[P]$ under $P$. Therefore, $\mu$ is not Pareto dominated by $\varphi^k[P]$.

Fix a student $i\in I$, and let $\eta=\varphi^k[P_i,P^\mu_{-i}]$. Since $\varphi^k$ is constrained stable, $\eta$ is stable under $(P_i,P^\mu_{-i})^k=(P_i^k,P^\mu_{-i})$. Define an auxiliary mechanism $\widehat\varphi_i$ that selects $\eta$ at $(P_i^k,P^\mu_{-i})$ and selects the outcome of $\DA$ at any other preference profile. By construction, $\widehat\varphi_i$ is stable.
Then, applying \Cref{lm:DA-stable} to the stable matching $\mu$ under $P^k$ gives $\eta = \widehat\varphi_i[P_i^k,P^\mu_{-i}]=\mu$. Thus, $\mu$ is a safe deviation of $\varphi^k$ under $P$ and $\varphi^k$ is not credible. Therefore, $\DA^k$ is the only constrained stable and credible mechanism.

For part~(iv), the result follows from \Cref{thm:DA-TTC}~(ii), \Cref{lm:constrained}~(ii), and \Cref{lm:simple}~(i).
The strict relations come from the example in the proof of \Cref{thm:DA-TTC}.
\end{proof}

\subsection{Proof of Section~\ref{sec:weak}}

\subsubsection{Proof of Proposition~\ref{prop:weak}}

\begin{proof}[Proof of \Cref{prop:weak}]
We start by proving part~(i). Fix a stable mechanism $\varphi$ and a weak-priority problem $[P,\succsim,\succ]$ under which $\DA$ has a safe deviation $\mu$. We show that $\mu$ is a safe deviation of $\varphi$.

Fix a student $i\in I$. By condition~(ii) of \Cref{def:safe-weak}, there exist a preference profile $P'\in \mathcal{P}$ and a tie-breaking $\succ'$ of $\succsim$ such that $o_i(\DA[P_i,P'_{-i},\succ'],\Gamma)=o_i(\mu,\Gamma)$ and $\succ'$ gives $i$ rank $r_{is}$ at every school $s\in S$. Let $\eta=\DA[P_i,P'_{-i},\succ']$. Since $\eta$ is stable under $(P_i,P'_{-i})$ and $\succ'$, \Cref{lm:DA-stable} gives $\varphi[P_i,P^\eta_{-i},\succ']=\eta$. Thus, $o_i(\varphi[P_i,P^\eta_{-i},\succ'],\Gamma)=o_i(\eta,\Gamma)=o_i(\mu,\Gamma)$. The tie-breaking $\succ'$ preserves $i$'s observed rank at every school.

It remains to verify conditions~(i) and~(iii) of \Cref{def:safe-weak}.
Suppose first that $\mu=\varphi[P,\succ]$. Then, the student optimality of $\DA$ and $\mu\neq\DA[P,\succ]$ imply that $\DA[P,\succ]$ Pareto dominates $\varphi[P,\succ]=\mu$, contradicting condition~(iii) of \Cref{def:safe-weak}. Thus, $\mu\neq\varphi[P,\succ]$.
Suppose second that $\mu$ is Pareto dominated by $\varphi[P,\succ]$. Then, the student optimality of $\DA$ implies that $\DA[P,\succ]$ Pareto dominates $\mu$, again contradicting condition~(iii) of \Cref{def:safe-weak}. 
Therefore, $\mu$ is a safe deviation of $\varphi$ and $\DA$ is more credible than $\varphi$ under weak priorities.

We next prove part~(ii). Fix a weak-priority problem $[P,\succsim,\succ]$ under which $\DA$ has a safe deviation $\mu$. We consider two cases.

Suppose first that $\BM[P,\succ]$ is stable under $P$ and $\succ$. Since $\DA[P,\succ]$ is the student-optimal stable matching and $\BM$ is Pareto efficient, $\BM[P,\succ]=\DA[P,\succ]$. Thus, $\mu\neq\BM[P,\succ]$, and $\mu$ is not Pareto dominated by $\BM[P,\succ]$. 
Fix a student $i\in I$. By condition~(ii) of \Cref{def:safe-weak}, there exist a preference profile $P'\in \mathcal{P}$ and a tie-breaking $\succ'$ of $\succsim$ such that $o_i(\DA[P_i,P'_{-i},\succ'],\Gamma)=o_i(\mu,\Gamma)$ and $\succ'$ gives $i$ rank $r_{is}$ at every school $s\in S$. Let $\eta=\DA[P_i,P'_{-i},\succ']$. Since $\eta$ is stable under $(P_i,P'_{-i})$ and $\succ'$, \Cref{lm:DA-BM} gives $\BM[P_i,P^\eta_{-i},\succ']=\eta$. Thus, $o_i(\BM[P_i,P^\eta_{-i},\succ'],\Gamma)=o_i(\eta,\Gamma)=o_i(\mu,\Gamma)$. The tie-breaking $\succ'$ preserves $i$'s observed rank at every school. Therefore, $\mu$ is a safe deviation of $\BM$.

Suppose second that $\BM[P,\succ]$ is not stable under $P$ and $\succ$. By \Cref{lm:BM-unstable}, $\BM$ has a safe deviation under $[P,\succ]$. Since $\succ$ is a tie-breaking of $\succsim$ and gives any student her observed rank at every school, any innocent explanation under $[P,\succ]$ is also an innocent explanation under $[P,\succsim,\succ]$. Therefore, $\BM$ has a safe deviation under $[P,\succsim,\succ]$. We conclude that $\DA$ is more credible than $\BM$ under weak priorities.

The strict relation comes from the example in the proof of \Cref{thm:DA-BM}.

Finally, we prove part~(iii). Suppose that $\Gamma(i)=\{i\}$ for any $i\in I$. We first compare $\TTC$ and $\DA$. Fix a weak-priority problem $[P,\succsim,\succ]$ under which $\TTC$ has a safe deviation $\mu$. We consider two cases.

Suppose first that $\TTC[P,\succ]=\DA[P,\succ]$. We claim that $\mu$ is a safe deviation of $\DA$ under $[P,\succsim,\succ]$. 
Since $\TTC[P,\succ]=\DA[P,\succ]$, we have $\mu\neq \DA[P,\succ]$ and $\mu$ is not Pareto dominated by $\DA[P,\succ]$. 
Fix a student $i\in I$. By condition~(ii) of \Cref{def:safe-weak}, there exist a preference profile $P'\in \mathcal{P}$ and a tie-breaking $\succ'$ of $\succsim$ such that $o_i(\TTC[P_i,P'_{-i},\succ'],\Gamma)=o_i(\mu,\Gamma)$ and $\succ'$ gives $i$ rank $r_{is}$ at every school $s\in S$.
Combining with $\Gamma(i)=\{i\}$, we have $\TTC[P_i,P'_{-i},\succ'](i)=\mu(i)$. Then, \Cref{lm:DA-TTC} implies that there exists $P''_{-i}$ such that $\DA[P_i,P''_{-i},\succ'](i)=\TTC[P_i,P'_{-i},\succ'](i)=\mu(i)$. Thus, $o_i(\DA[P_i,P''_{-i},\succ'],\Gamma)=o_i(\mu,\Gamma)$. Moreover, the tie-breaking $\succ'$ preserves $i$'s observed rank at every school. Thus, $\mu$ is a safe deviation of $\DA$.

Suppose second that $\TTC[P,\succ]\neq\DA[P,\succ]$. Let $\eta=\TTC[P,\succ]$. We claim that $\eta$ is a safe deviation of $\DA$ under $[P,\succsim,\succ]$. Indeed, the Pareto efficiency of $\TTC$ ensures that $\eta$ is not Pareto dominated by $\DA[P,\succ]$. Moreover, for any student $i\in I$, \Cref{lm:DA-TTC} implies that there exists $P'_{-i}$ such that $\DA[P_i,P'_{-i},\succ](i)=\eta(i)$. Since $\Gamma(i)=\{i\}$, $o_i(\DA[P_i,P'_{-i},\succ],\Gamma)=o_i(\eta,\Gamma)$. The tie-breaking $\succ$ preserves $i$'s observed rank at every school. Thus, $\eta$ is a safe deviation of $\DA$. We conclude that $\TTC$ is more credible than $\DA$ under weak priorities.

We next compare $\EA$ and $\DA$. Fix a weak-priority problem $[P,\succsim,\succ]$ under which $\EA$ has a safe deviation $\mu$. We show that $\mu$ is also a safe deviation of $\DA$ under $[P,\succsim,\succ]$. 

First, we show that $\DA[P,\succ]\neq \mu$. Otherwise, since $\EA$ weakly Pareto dominates $\DA$ and $\EA[P,\succ]\neq \DA[P,\succ]$, $\EA[P,\succ]$ Pareto dominates $\DA[P,\succ]=\mu$, contradicting condition~(iii) of \Cref{def:safe-weak}. Thus, $\DA[P,\succ]\neq \mu$.

Next, we show that $\mu$ is not Pareto dominated by $\DA[P,\succ]$. Otherwise, since $\EA$ weakly Pareto dominates $\DA$, $\EA[P,\succ]$ Pareto dominates $\mu$, contradicting condition~(iii) of \Cref{def:safe-weak}. Thus, $\mu$ is not Pareto dominated by $\DA[P,\succ]$. 

Fix a student $i\in I$. By condition~(ii) of \Cref{def:safe-weak}, there exist a preference profile $P'\in \mathcal{P}$ and a tie-breaking $\succ'$ of $\succsim$ such that $o_i(\EA[P_i,P'_{-i},\succ'],\Gamma)=o_i(\mu,\Gamma)$ and $\succ'$ gives $i$ rank $r_{is}$ at every school $s\in S$. Since $\Gamma(i)=\{i\}$, this implies that $\EA[P_i,P'_{-i},\succ'](i)=\mu(i)$. Then, \Cref{lm:DA-TTC} implies that there exists $P''_{-i}$ such that $\DA[P_i,P''_{-i},\succ'](i)=\EA[P_i,P'_{-i},\succ'](i)=\mu(i)$. Therefore, $o_i(\DA[P_i,P''_{-i},\succ'],\Gamma)=o_i(\mu,\Gamma)$. The tie-breaking $\succ'$ preserves $i$'s observed rank at every school. Thus, $\mu$ is a safe deviation of $\DA$, and therefore $\EA$ is more credible than $\DA$ under weak priorities.

The strict relations come from the example in the proof of \Cref{thm:DA-TTC}.
\end{proof}

\subsubsection{Proof of Proposition~\ref{prop:weak-indifferent-DA-TTC}}

\begin{lm}\label{lm:weak-indifferent-DA-TTC}
Suppose that all schools are indifferent among all students under $\succsim$. Fix a preference profile $P$, a student $i\in I$, a tie-breaking $\succ$ of $\succsim$ under which $i$ has rank $r_{is}$ at any school $s\in S$, and a mechanism $\varphi\in\{\TTC,\EA\}$. Let $\mu=\varphi[P,\succ]$. Then:
\begin{itemize}
  \item[(i)] $\mu(i)\mathrel{R_i}\emptyset$, and $|\mu(s)|=q_s<r_{is}$ for any $s\mathrel{P_i}\mu(i)$.
  \item[(ii)] There exist $\widehat P_{-i}$ and a tie-breaking $\widehat\succ$ of $\succsim$ such that $\DA[P_i,\widehat P_{-i},\widehat\succ]=\mu$ and $i$ has rank $r_{is}$ at any school $s\in S$ under $\widehat\succ$.
\end{itemize}
\end{lm}

\begin{proof}[Proof of \Cref{lm:weak-indifferent-DA-TTC}]
We first prove part~(i). 
Since $\TTC$ and $\EA$ are individually rational, $\mu(i)\mathrel{R_i}\emptyset$. 
Moreover, since $\TTC$ and $\EA$ are Pareto efficient, $|\mu(s)|=q_s$ for any $s\mathrel{P_i}\mu(i)$.
By \Cref{lm:DA-TTC}, there exists $P'_{-i}$ such that $\DA[P_i,P'_{-i},\succ](i)=\mu(i)$. Since $\DA[P_i,P'_{-i},\succ]$ is stable and $s\mathrel{P_i}\DA[P_i,P'_{-i},\succ](i)$, school $s$ assigns all its $q_s$ seats to students ranked above $i$ under $\succ_s$. Thus, $q_s\leq r_{is}-1$, and therefore $q_s<r_{is}$.

We next prove part~(ii). Construct a strict priority profile $\widehat\succ$ as follows. At any school $s\mathrel{P_i}\mu(i)$, place all students in $\mu(s)$ above $i$ while keeping $i$ in position $r_{is}$. This construction can be done since $q_s<r_{is}$. Complete the remaining priorities arbitrarily while preserving $i$'s rank. Since all schools are indifferent among all students under $\succsim$, $\widehat\succ$ is a tie-breaking of $\succsim$.

We claim that $\mu$ is stable under $(P_i,P^\mu_{-i})$ and $\widehat\succ$.
Indeed, by part~(i), $\mu$ is individually rational for $i$ and no blocking pair involves $i$. Moreover, for any other student $j \neq i$, the definition of $P^\mu$ ensures that $\mu$ is individually rational for $j$ and no blocking pair involves $j$.
Therefore, $\mu$ is stable under $(P_i,P^\mu_{-i})$ and $\widehat\succ$.
Then, \Cref{lm:DA-stable} gives $\DA[P_i,P^\mu_{-i},\widehat\succ]=\mu$. Thus, part~(ii) holds with $\widehat P_{-i}=P^\mu_{-i}$.
\end{proof}

\begin{proof}[Proof of \Cref{prop:weak-indifferent-DA-TTC}]
Fix $\varphi\in\{\TTC,\EA\}$. We first show that $\varphi$ is more credible than $\DA$. Fix a weak-priority problem $[P,\succsim,\succ]$ under which $\varphi$ has a safe deviation $\mu$. We consider two cases.

Suppose first that $\varphi[P,\succ]=\DA[P,\succ]$. Since $\mu$ is a safe deviation of $\varphi$, we have $\mu\neq\DA[P,\succ]$, and $\mu$ is not Pareto dominated by $\DA[P,\succ]$. Fix a student $i\in I$. By condition~(ii) of \Cref{def:safe-weak}, there exist a preference profile $P'\in \mathcal{P}$ and a tie-breaking $\succ'$ of $\succsim$ such that $o_i(\varphi[P_i,P'_{-i},\succ'],\Gamma)=o_i(\mu,\Gamma)$ and $\succ'$ gives $i$ rank $r_{is}$ at every school $s\in S$. By part~(ii) of \Cref{lm:weak-indifferent-DA-TTC}, there exist $\widehat P_{-i}$ and a tie-breaking $\widehat\succ$ such that $\DA[P_i,\widehat P_{-i},\widehat\succ]=\varphi[P_i,P'_{-i},\succ']$ and $\widehat\succ$ preserves $i$'s observed rank at every school. Thus, $o_i(\DA[P_i,\widehat P_{-i},\widehat\succ],\Gamma)=o_i(\mu,\Gamma)$. Therefore, $\mu$ is a safe deviation of $\DA$.

Suppose second that $\varphi[P,\succ]\neq\DA[P,\succ]$. Fix a student $i\in I$. 
By part~(ii) of \Cref{lm:weak-indifferent-DA-TTC}, there exist $\widehat P_{-i}$ and a tie-breaking $\widehat\succ$ such that $\DA[P_i,\widehat P_{-i},\widehat\succ]=\varphi[P,\succ]$ and $\widehat\succ$ preserves $i$'s observed rank at every school.
Thus, $o_i(\DA[P_i,\widehat P_{-i},\widehat\succ],\Gamma)=o_i(\varphi[P,\succ],\Gamma)$. Moreover, $\varphi[P,\succ]\neq\DA[P,\succ]$, and the Pareto efficiency of $\varphi$ implies that $\varphi[P,\succ]$ is not Pareto dominated by $\DA[P,\succ]$. Therefore, $\varphi[P,\succ]$ is a safe deviation of $\DA$.
We conclude that both $\TTC$ and $\EA$ are more credible than $\DA$.

To show that $\TTC$ and $\EA$ are strictly more credible than $\DA$, consider the following example. 
Consider three students $\{i_1,i_2,i_3\}$ and two schools $\{s_1,s_2\}$. Let $q_{s_1}=q_{s_2}=1$. Suppose that every school is indifferent among all students under $\succsim$. Consider the following preference profile and realized tie-breaking:
\begin{table}[H]
\centering
\begin{tabular}{lll|ll}
  $P_{i_1}$ & $P_{i_2}$ & $P_{i_3}$ & $\succ_{s_1}$ & $\succ_{s_2}$ \\ \hline
  $s_1$ & $s_1$ & $s_2$ & $i_3$ & $i_1$ \\
  $s_2$ & $s_2$ & $s_1$ & $i_1$ & $i_2$ \\
  $\emptyset$ & $\emptyset$ & $\emptyset$ & $i_2$ & $i_3$
\end{tabular}
\end{table}
Consider the matchings
\[
\mu=
\begin{pmatrix}
i_1 & i_2 & i_3\\
s_2 & \emptyset & s_1
\end{pmatrix},
\qquad
\eta=
\begin{pmatrix}
i_1 & i_2 & i_3\\
s_1 & \emptyset & s_2
\end{pmatrix}.
\]
Notice that $\DA[P, \succ] = \mu$ and $\TTC[P, \succ] = \EA[P, \succ] = \eta$.

First, we show that neither $\TTC$ nor $\EA$ has a safe deviation under this problem, regardless of the observation structure. Notice that it is sufficient to consider $\Gamma(i)=\{i\}$ for any $i\in I$. Fix $\varphi\in\{\TTC,\EA\}$, and suppose that $\nu$ is a safe deviation. If $\nu(i_1)=\emptyset$, by part~(i) of \Cref{lm:weak-indifferent-DA-TTC}, $1=q_{s_2}<r_{i_1s_2}=1$, a contradiction. Thus, $\nu(i_1)\in\{s_1,s_2\}$.
Similarly, if $\nu(i_3)=\emptyset$, by part~(i) of \Cref{lm:weak-indifferent-DA-TTC}, $1=q_{s_1}<r_{i_3s_1}=1$, a contradiction. Thus, $\nu(i_3)\in\{s_1,s_2\}$.
Feasibility therefore requires $\nu$ to be either $\mu$ or $\eta$. Then, condition~(i) of \Cref{def:safe-weak} rules out $\eta=\varphi[P,\succ]$, while condition~(iii) rules out $\mu$ since $\eta$ Pareto dominates $\mu$. Therefore, $\varphi$ has no safe deviation.

Second, we show that $\eta$ is a safe deviation of $\DA$. For any $i\in I$, by part~(ii) of \Cref{lm:weak-indifferent-DA-TTC}, there exist $\widehat P_{-i}$ and a tie-breaking $\widehat\succ$ of $\succsim$ such that $\DA[P_i,\widehat P_{-i},\widehat\succ]=\eta$ and $\widehat\succ$ preserves $i$'s observed rank at every school. Thus, $o_i(\DA[P_i,\widehat P_{-i},\widehat\succ], \Gamma) = o_i(\eta, \Gamma)$ under any observation structure $\Gamma$. Moreover, $\eta\neq\mu=\DA[P,\succ]$, and $\eta$ Pareto dominates $\mu$, so $\eta$ is not Pareto dominated by $\mu$. Therefore, $\eta$ is a safe deviation of $\DA$.
\end{proof}

\subsection{Proof of Section~\ref{sec:affirmative}}

\subsubsection{Proof of Proposition~\ref{prop:affirmative-full-observation}}

\begin{proof}[Proof of \Cref{prop:affirmative-full-observation}]
Suppose that $\Gamma(i) = I$ for any $i \in I$.
Fix a policy $\alpha\in\{r,q\}$, and suppose, toward a contradiction, that $\mu$ is a safe deviation of $\DA^\alpha$ under $P$. Fix a student $i\in I$. By condition~(ii) of \Cref{def:safe}, there exists a preference profile $P' \in \mathcal{P}$ such that $\DA^\alpha[P_i,P'_{-i}]=\mu$.

We first show that $\mu$ is stable under $P$ for policy $\alpha$. Since $\DA^\alpha$ is individually rational, $\mu(i)\mathrel{R_i}\emptyset$. Thus, $\mu$ is individually rational under $P$. Moreover, if $(i,s)$ forms a blocking pair for $\mu$ under $P$ and policy $\alpha$, then $(i,s)$ also blocks $\mu$ under $(P_i,P'_{-i})$, contradicting the policy-specific stability of $\DA^\alpha[P_i,P'_{-i}]$ \citep{hafalir2013effective}. Thus, $\mu$ is stable under $P$ for policy $\alpha$.

Since $\DA^\alpha$ selects the student-optimal matching among the matchings that are stable under policy $\alpha$ \citep{hafalir2013effective}, $\DA^\alpha[P]$ weakly Pareto dominates $\mu$. Since $\mu\neq\DA^\alpha[P]$, $\DA^\alpha[P]$ Pareto dominates $\mu$, contradicting condition~(iii) of \Cref{def:safe}. Therefore, $\DA^\alpha$ is credible.
\end{proof}

\subsubsection{Proof of Proposition~\ref{prop:affirmative-reserve-quota}}

\begin{lm}\label{lm:reserve-quota}
Suppose $r_s+q_s^M=q_s$ and $q_s^M>0$ for any $s\in S$. Fix a preference profile $P$ and a student $i\in I$. There exists $\widehat P_{-i}$ such that $\DA^q[P_i,\widehat P_{-i}](j)=\DA^r[P](j)$ for any $j\in I_m\cup\{i\}$.
\end{lm}

\begin{proof}[Proof of \Cref{lm:reserve-quota}]
Let $\mu=\DA^r[P]$. Construct $\widehat P_{-i}$ as follows. For any student $j\neq i$, let $\widehat P_j=P^\mu_j$ if $j\in I_m$ or $\mu(j)\mathrel{P_i}\mu(i)$. Any other student reports no school as acceptable.
Let $\eta=\DA^q[P_i,\widehat P_{-i}]$.

We first show that $i$ is rejected by any school that she prefers to $\mu(i)$ and any minority student in $\mu(s)$ remains assigned to $s$ under $\eta$. 
Fix a school $s$ such that $s\mathrel{P_i}\mu(i)$, and let $m_s=|\mu(s)\cap I_m|$. Since $\mu$ is stable under minority reserves, $|\mu(s)|=q_s$.

We consider three cases.
Suppose first that $i\in I_m$. Stability under minority reserves implies that $m_s\geq r_s$ and that any student in $\mu(s)$ has higher priority than $i$. Moreover, $|\mu(s)\setminus I_m|=q_s-m_s\leq q_s-r_s=q_s^M$. Thus, $\mu(s)$ is feasible under majority quotas.
Under $(P_i, \widehat P_{-i})$, only $i$ and the students in $\mu(s)$ find $s$ acceptable, and any student in $\mu(s)$ has higher priority than $i$. If some $j\in\mu(s)$ is not assigned to $s$ under $\eta$, then $(j,s)$ blocks $\eta$ under majority quotas, contradicting the stability of $\DA^q$ under majority quotas. Therefore, $\eta(s)=\mu(s)$ and $\eta(i)\neq s$.

Suppose next that $i\notin I_m$ and $m_s>r_s$. Since $|\mu(s)\setminus I_m|=q_s-m_s<q_s-r_s=q_s^M$, $\mu(s)$ is feasible under majority quotas. Moreover, stability under minority reserves implies that any student in $\mu(s)$ has higher priority than $i$. By the definition of $\widehat P_{-i}$ and stability under majority quotas, we have $\eta(s)=\mu(s)$ and $\eta(i)\neq s$.

Suppose finally that $i\notin I_m$ and $m_s\leq r_s$. There are $|\mu(s)\setminus I_m|=q_s-m_s$ majority students in $\mu(s)$, all of whom have higher priority than $i$ by stability under minority reserves. Since $q_s-m_s \geq q_s-r_s=q_s^M$, stability under majority quotas implies that $s$ admits $q_s^M$ majority students under $\eta$. 
We first claim that $\eta(i) \neq s$. Otherwise, at least one majority student in $\mu(s)$ with higher priority than $i$ is unassigned, contradicting stability under majority quotas. 
We then claim that any minority student in $\mu(s)$ is assigned to $s$ under $\eta$. Suppose, toward a contradiction, that some $j\in\mu(s)\cap I_m$ is not assigned to $s$. Then $|\eta(s)|\leq q_s^M+m_s-1\leq q_s^M+r_s-1=q_s-1$, and thus $(j,s)$ blocks $\eta$, contradicting stability under majority quotas. Therefore, in all three cases, $i$ is rejected by $s$ and any minority student in $\mu(s)$ remains assigned to $s$.

We next show that $\eta(i)=\mu(i)$. If $\mu(i)=\emptyset$, the preceding argument and individual rationality give $\eta(i)=\emptyset$. Suppose instead that $\mu(i)\in S$. By the definition of $\widehat P_{-i}$, the only students other than $i$ who find $\mu(i)$ acceptable are the minority students assigned there under $\mu$. There are at most $q_{\mu(i)}-1$ such students. If $i\in I_m$, student $i$ can therefore fill an available seat at $\mu(i)$. If $i\notin I_m$, student $i$ can fill an available majority position since $q^M_{\mu(i)}>0$. Since the preceding argument rules out any school that $i$ prefers to $\mu(i)$, stability under majority quotas implies that $\eta(i)=\mu(i)$.

Finally, we show that $\eta(j)=\mu(j)$ for any $j\in I_m\setminus\{i\}$. Fix a student $j\in I_m\setminus\{i\}$. If $\mu(j)=\emptyset$, then $\widehat P_j$ lists no school as acceptable, so individual rationality gives $\eta(j)=\emptyset$. 
If $\mu(j)=s\in S$ and $s\mathrel{P_i}\mu(i)$, the preceding argument gives $\eta(j)=s$. 
If $\mu(j)=s\in S$ and $\mu(i)=s$, we claim that $\eta(j)=s$. Otherwise, since at most $q_s-1$ students other than $j$ find $s$ acceptable under $(P_i,\widehat P_{-i})$, $(j,s)$ blocks $\eta$, contradicting stability under majority quotas. 
If $\mu(j)=s\in S$ and $\mu(i)\mathrel{P_i}s$, we claim that $\eta(j)=s$. Otherwise, since at most $q_s-1$ students other than $i$ and $j$ find $s$ acceptable under $(P_i,\widehat P_{-i})$ and $\eta(i)=\mu(i) \neq s$, $(j,s)$ blocks $\eta$, contradicting stability under majority quotas. 
Therefore, $\eta(j)=\mu(j)$ for any $j\in I_m\setminus\{i\}$. Together with $\eta(i)=\mu(i)$, we conclude that $\DA^q[P_i,\widehat P_{-i}](j)=\DA^r[P](j)$ for any $j\in I_m\cup\{i\}$.
\end{proof}

\begin{proof}[Proof of \Cref{prop:affirmative-reserve-quota}]
Suppose that $\mu$ is a safe deviation of $\DA^r$ under $[P,\succ,q,I_m,r,q^M]$. Fix a student $i\in I$. By condition~(ii) of \Cref{def:safe}, there exists $P'\in \mathcal{P}$ such that $o_i(\DA^r[P_i,P'_{-i}],\Gamma)=o_i(\mu,\Gamma)$. Then, by \Cref{lm:reserve-quota}, there exists $\widehat P_{-i}$ such that $\DA^q[P_i,\widehat P_{-i}](j)=\DA^r[P_i,P'_{-i}](j)$ for any $j\in I_m\cup\{i\}$. Since $\Gamma(i)\subseteq I_m\cup\{i\}$, $o_i(\DA^q[P_i,\widehat P_{-i}],\Gamma)=o_i(\DA^r[P_i,P'_{-i}],\Gamma)=o_i(\mu,\Gamma)$.

It remains to verify conditions~(i) and~(iii) of \Cref{def:safe}. We first claim that $\mu\neq\DA^q[P]$. Otherwise, since $\DA^r[P]$ weakly Pareto dominates $\DA^q[P]$ \citep{hafalir2013effective} and $\mu\neq\DA^r[P]$, $\DA^r[P]$ Pareto dominates $\mu$, contradicting that $\mu$ is a safe deviation of $\DA^r$ under $P$.

We next claim that $\mu$ is not Pareto dominated by $\DA^q[P]$. Otherwise, since $\DA^r[P]$ weakly Pareto dominates $\DA^q[P]$, $\DA^r[P]$ Pareto dominates $\mu$, contradicting that $\mu$ is a safe deviation of $\DA^r$ under $P$.
Therefore, $\mu$ is a safe deviation of $\DA^q$ under $P$, and $\DA^r$ is more credible than $\DA^q$.

To show that $\DA^r$ is strictly more credible than $\DA^q$, consider the following example.
Consider three students $\{i_1,i_2,i_3\}$ and two schools $\{s_1,s_2\}$. Let $q_{s_1}=2$, $r_{s_1}=q_{s_1}^M=1$, $q_{s_2}=q_{s_2}^M=1$, and $r_{s_2}=0$. Let $I_m=\{i_1\}$. Consider the following preference profile and priority profile:
\begin{table}[H]
\centering
\begin{tabular}{lll|l}
  $P_{i_1}$ & $P_{i_2}$ & $P_{i_3}$ & $\succ_{s_1}$ \\ \hline
  $s_2$ & $s_1$ & $s_1$ & $i_1$ \\
  $\emptyset$ & $\emptyset$ & $\emptyset$ & $i_2$ \\
  &  &  & $i_3$
\end{tabular}
\end{table}
The priority of $s_2$ can be arbitrary. Consider the matchings
\[
\mu=
\begin{pmatrix}
i_1 & i_2 & i_3\\
s_2 & s_1 & s_1
\end{pmatrix},
\qquad
\eta=
\begin{pmatrix}
i_1 & i_2 & i_3\\
s_2 & s_1 & \emptyset
\end{pmatrix}.
\]
Notice that $\DA^r[P]=\mu$ and $\DA^q[P]=\eta$.

First, we show that $\DA^r$ has no safe deviation under any observation structure. Indeed, every student receives her most-preferred outcome under $\mu$, so $\mu$ Pareto dominates any different matching.

Second, we show that $\mu$ is a safe deviation of $\DA^q$ under any observation structure satisfying $\Gamma(i)\subseteq I_m\cup\{i\}$ for any $i\in I$. By \Cref{lm:reserve-quota}, for any $i\in I$ there exists $\widehat P_{-i}$ such that $\DA^q[P_i,\widehat P_{-i}]$ agrees with $\mu$ on $I_m\cup\{i\}$. Since $\Gamma(i)\subseteq I_m\cup\{i\}$, we have $o_i(\DA^q[P_i,\widehat P_{-i}],\Gamma)=o_i(\mu,\Gamma)$. Moreover, $\mu\neq\eta$, and $\eta$ does not Pareto dominate $\mu$. Therefore, $\mu$ is a safe deviation of $\DA^q$.
\end{proof}

\section{Omitted examples}\label[appendix]{sec:example}

\begin{example}\label{exp:school-proposing}
In this example, we show that the student-proposing deferred acceptance mechanism ($\DA$) is strictly more credible than the school-proposing deferred acceptance mechanism ($\DAS$) under any observation structure.

Consider three students $\{i_1,i_2,i_3\}$ and two schools $\{s_1,s_2\}$. Let $q_{s_1}=q_{s_2}=1$. Consider the following preference profile and priority profile:
\begin{table}[H]
\centering
\begin{tabular}{lll|ll}
  $P_{i_1}$ & $P_{i_2}$ & $P_{i_3}$ & $\succ_{s_1}$ & $\succ_{s_2}$ \\ \hline
  $s_2$ & $s_1$ & $\emptyset$ & $i_1$ & $i_2$ \\
  $s_1$ & $s_2$ &  & $i_2$ & $i_1$ \\
  $\emptyset$ & $\emptyset$ &  & $i_3$ & $i_3$
\end{tabular}
\end{table}
Consider the matchings
\[
\mu=\begin{pmatrix}
i_1 & i_2 & i_3 \\
s_2 & s_1 & \emptyset
\end{pmatrix},
\qquad
\eta=\begin{pmatrix}
i_1 & i_2 & i_3 \\
s_1 & s_2 & \emptyset
\end{pmatrix}.
\]
Notice that $\DA[P]=\mu$ and $\DAS[P]=\eta$.

First, we show that $\DA$ has no safe deviation under $P$ for any observation structure. Indeed, every student receives her most-preferred outcome under $\mu$, so $\mu$ Pareto dominates any different matching.

Second, we show that $\mu$ is a safe deviation of $\DAS$ under $P$ for any observation structure. Since $\mu$ is stable and $\DAS$ is a stable mechanism, \Cref{lm:DA-stable} gives $\DAS[P_i,P^\mu_{-i}]=\mu$ for any $i\in I$. Thus, $o_i(\DAS[P_i,P^\mu_{-i}],\Gamma)=o_i(\mu,\Gamma)$. Moreover, $\mu\neq\DAS[P]=\eta$, and $\eta$ does not Pareto dominate $\mu$. Therefore, $\mu$ is a safe deviation of $\DAS$ under $P$.
\end{example}

\begin{example}\label{exp:BM-TTC}
In this example, we show that $\BM$ is incomparable with either $\TTC$ or $\EA$ when students observe part, but not all, of the assignment. 

In both problems below, let $\Gamma(i_1)=\Gamma(i_2)=\{i_1,i_2\}$ and $\Gamma(i_3)=\{i_3\}$.

First, to show that $\BM$ is not more credible than either $\TTC$ or $\EA$, consider the example in the proof of \Cref{thm:DA-BM}. Notice that $\TTC[P]=\EA[P]=\mu$. Also, as shown there, $\DA[P]=\mu$, $\BM[P]=\eta$, and $\mu$ is a safe deviation of $\BM$ under any observation structure.

We show that neither $\TTC$ nor $\EA$ has a safe deviation under $P$. Fix $\varphi\in\{\TTC,\EA\}$, and suppose, toward a contradiction, that $\nu$ is a safe deviation of $\varphi$. For any $i\in I$, since $i \in \Gamma(i)$, condition~(ii) of \Cref{def:safe} gives $\widehat P \in \mathcal{P}$ such that $\varphi[P_i, \widehat P_{-i}](i) = \nu(i)$. Then, by \Cref{lm:DA-TTC}, there exists $P'_{-i}$ such that $\DA[P_i, P'_{-i}](i)=\nu(i)$.
For any $P'_{-i_1}$, since $i_1$ ranks $s_1$ first and has the highest priority there, stability implies that $\nu(i_1) = \DA[P_{i_1},P'_{-i_1}](i_1)=s_1$. Moreover, since $i_2$ is ranked second by both schools, stability implies that $\nu(i_2) = \DA[P_{i_2},P'_{-i_2}](i_2)\in\{s_1,s_2\}$.
Then, feasibility gives $\nu(i_2)=s_2$ and $\nu(i_3)=\emptyset$. This implies that $\nu=\mu=\varphi[P]$, a contradiction. Therefore, neither $\TTC$ nor $\EA$ has a safe deviation under $P$.

Second, to show that neither $\TTC$ nor $\EA$ is more credible than $\BM$, consider three students $\{i_1,i_2,i_3\}$ and two schools $\{s_1,s_2\}$. Let $q_{s_1}=q_{s_2}=1$. Consider the following preference profile and priority profile:
\begin{table}[H]
\centering
\begin{tabular}{lll|ll}
  $P_{i_1}$ & $P_{i_2}$ & $P_{i_3}$ & $\succ_{s_1}$ & $\succ_{s_2}$ \\ \hline
  $s_2$ & $s_1$ & $s_1$ & $i_1$ & $i_3$ \\
  $s_1$ & $s_2$ & $\emptyset$ & $i_2$ & $i_1$ \\
  $\emptyset$ & $\emptyset$ &  & $i_3$ & $i_2$
\end{tabular}
\end{table}
Consider the matchings
\[
\mu=
\begin{pmatrix}
i_1 & i_2 & i_3\\
s_2 & \emptyset & s_1
\end{pmatrix},
\qquad
\eta=
\begin{pmatrix}
i_1 & i_2 & i_3\\
s_2 & s_1 & \emptyset
\end{pmatrix}.
\]
Notice that $\TTC[P]=\mu$ and $\BM[P]=\EA[P]=\eta$.

First, we show that $\eta$ is a safe deviation of $\TTC$ under $P$. Let $P'_{i_3}$ report no school as acceptable. Then $\TTC[P_{i_1},P_{i_2},P'_{i_3}]=\eta$, so this profile explains the observations of $i_1$ and $i_2$. For $i_3$, let $P'_{i_1}$ rank $s_1$ first and $s_2$ second. Then $\TTC[P'_{i_1},P_{i_2},P_{i_3}](i_3)=\emptyset=\eta(i_3)$, so this profile explains $i_3$'s observation. Moreover, $\eta\neq\TTC[P]=\mu$, and $\mu$ does not Pareto dominate $\eta$. Thus, $\eta$ is a safe deviation of $\TTC$ under $P$.

Second, we show that $\mu$ is a safe deviation of $\EA$ under $P$. Let $P'_{i_3}$ rank $s_1$, $s_2$, and $\emptyset$ in that order. Then $\EA[P_{i_1},P_{i_2},P'_{i_3}]=\mu$, so this profile explains the observations of $i_1$ and $i_2$. For $i_3$, let $i_1$ report $P_{i_1}$ and let $i_2$ report $s_2$ as her only acceptable school. Then $\EA[P_{i_1},P'_{i_2},P_{i_3}]=\mu$, so this profile explains $i_3$'s observation. Moreover, $\mu\neq\EA[P]=\eta$, and $\eta$ does not Pareto dominate $\mu$. Therefore, $\mu$ is a safe deviation of $\EA$ under $P$.

Finally, we show that $\BM$ has no safe deviation under $P$. Suppose, toward a contradiction, that $\nu$ is a safe deviation of $\BM$. Since $\eta$ assigns $i_1$ and $i_2$ to their most-preferred schools and $i_3$ prefers $\emptyset$ to $s_2$, $\eta$ Pareto dominates any matching different from $\eta$ unless $i_3$ is assigned to $s_1$. Thus, condition~(iii) of \Cref{def:safe} requires $\nu(i_3)=s_1$.

Since $q_{s_1}=q_{s_2}=1$, $(\nu(i_1),\nu(i_2))\in\{(s_2,\emptyset),(\emptyset,s_2),(\emptyset,\emptyset)\}$. We consider three cases. Suppose first that $(\nu(i_1),\nu(i_2))=(s_2,\emptyset)$. Under any profile $(P_{i_2},P'_{-i_2})$, student $i_2$ applies to $s_1$ in the first round. Since only $i_1$ has higher priority than $i_2$ at $s_1$, leaving $i_2$ unmatched requires $i_1$ to be accepted by $s_1$. Student $i_1$ therefore cannot be assigned to $s_2$, so $i_2$ cannot explain her observation.

Suppose next that $(\nu(i_1),\nu(i_2))=(\emptyset,s_2)$. Under any profile $(P_{i_1},P'_{-i_1})$, student $i_1$ applies to $s_2$ in the first round. Since only $i_3$ has higher priority than $i_1$ at $s_2$, leaving $i_1$ unmatched requires $i_3$ to be accepted by $s_2$. Student $i_2$ therefore cannot be assigned to $s_2$, so $i_1$ cannot explain her observation.

Finally, suppose that $(\nu(i_1),\nu(i_2))=(\emptyset,\emptyset)$. As in the first case, leaving $i_2$ unmatched requires $i_1$ to be accepted by $s_1$, contradicting $\nu(i_1)=\emptyset$. In every case, at least one student cannot explain her observation. Therefore, $\BM$ has no safe deviation under $P$.

Thus, under the second profile, both $\TTC$ and $\EA$ have safe deviations while $\BM$ has none. Together with the first profile, we conclude that $\BM$ is incomparable with either $\TTC$ or $\EA$ when students observe only part of the assignment.
\end{example}

\begin{example}\label{exp:DA-TTC-partial}
In this example, we show that $\DA$ is incomparable with either $\TTC$ or $\EA$ when students observe part, but not all, of the assignment.

In both problems below, let $\Gamma(i_1)=\Gamma(i_2)=\{i_1,i_2\}$ and $\Gamma(i_3)=\{i_3\}$.

First, to show that $\DA$ is not more credible than either $\TTC$ or $\EA$, consider the example used to prove part~(ii) of \Cref{thm:DA-TTC}. Recall that $\DA[P]=\mu$, $\TTC[P]=\EA[P]=\eta$, and neither $\TTC$ nor $\EA$ has a safe deviation when students observe only their own assignment.

We show that $\eta$ is a safe deviation of $\DA$ under the observation structure considered here. For any $i\in\{i_1,i_2\}$, $\DA[P_i,P^\eta_{-i}]=\eta$, so this profile explains the observations of $i_1$ and $i_2$. For $i_3$, let $i_1$ report no school as acceptable and let $i_2$ report only $s_2$ as acceptable. Then $\DA[P_{i_3},P'_{-i_3}](i_3)=\eta(i_3)$, so this profile explains $i_3$'s observation. Moreover, $\eta\neq\DA[P]=\mu$, and $\mu$ does not Pareto dominate $\eta$. Therefore, $\eta$ is a safe deviation of $\DA$.

We next show that neither $\TTC$ nor $\EA$ has a safe deviation under $P$. Fix $\varphi\in\{\TTC,\EA\}$, and suppose, toward a contradiction, that $\nu$ is a safe deviation of $\varphi$ under the observation structure considered here. Since $i\in\Gamma(i)$ for any $i\in I$, $\nu$ is also a safe deviation of $\varphi$ when students observe only their own assignment, contradicting the example in the proof of \Cref{thm:DA-TTC}. Therefore, neither $\TTC$ nor $\EA$ has a safe deviation under $P$.

Second, to show that neither $\TTC$ nor $\EA$ is more credible than $\DA$, consider the second profile in \Cref{exp:BM-TTC}. As shown there, both $\TTC$ and $\EA$ have safe deviations under $P$, while $\BM$ has none. Since \Cref{thm:DA-BM} shows that $\DA$ is more credible than $\BM$ under any observation structure, $\DA$ has no safe deviation under $P$.

Thus, under the first profile, $\DA$ has a safe deviation while neither $\TTC$ nor $\EA$ has one. Under the second profile, both $\TTC$ and $\EA$ have safe deviations while $\DA$ has none. Therefore, $\DA$ is incomparable with either $\TTC$ or $\EA$ when students observe only part of the assignment.
\end{example}

\bibliographystyle{apalike}
\bibliography{draft_ref}

\end{document}